\documentclass[acmsmall,nonacm]{acmart}

\usepackage{microtype}
\usepackage{amsthm,amsmath}
\usepackage{xspace}
\usepackage{setspace}  %
\usepackage[dvipsnames]{xcolor}
\usepackage{graphicx}
\usepackage[linesnumbered,ruled,vlined]{algorithm2e}

\usepackage[normalem]{ulem}
\usepackage{multirow}
\usepackage{subcaption}
\usepackage{cleveref}
\usepackage{todonotes}

\usepackage{xcolor}
\usepackage{colortbl}
\usepackage{enumerate}%
\usepackage{paralist}
\definecolor{jade}{rgb}{0.0, 0.66, 0.42}
\definecolor{cerise}{HTML}{CE4760}
\colorlet{fg}{jade!75!black}
\colorlet{bg}{cerise!75!black}
\colorlet{hl}{yellow!75!black}
\definecolor{mg}{RGB}{233,116,81}

\usepackage{amsmath}
\usepackage{forest}

\usepackage{bbding}
\theoremstyle{plain}  
\newtheorem{theorem}{Theorem}[section]
\newtheorem{lemma}[theorem]{Lemma}
\newtheorem{proposition}[theorem]{Proposition}
\newtheorem{corollary}[theorem]{Corollary}
\newtheorem{observation}[theorem]{Observation}

\newtheorem{conjecture}[theorem]{Conjecture}
\newtheorem{openquestion}[theorem]{Open Problem}
\newtheorem{claim}[theorem]{Claim}

\allowdisplaybreaks

\theoremstyle{definition}  
\newtheorem{definition}[theorem]{Definition}
\newtheorem{example}[theorem]{Example}

\theoremstyle{remark}

\newcommand{\alt}{\ensuremath{\mathsf{ALT}}\xspace}
\newcommand{\poly}{\ensuremath{\mathrm{poly}}\xspace}

\newcommand{\spn}{\ensuremath{\mathsf{span}}}

\newcommand{\ourclass}{\ensuremath{\mathsf{spanALP}}\xspace}

\title{From Alternation to FPRAS: Toward a Complexity Classification of Approximate Counting
}
\author{Markus Hecher}
\authornote{Supported by Austrian Science Fund (FWF J4656); French National Research Agency (ANR ANR-25-CE23-7647-01).}
\orcid{0000-0003-0131-6771}
\affiliation{
\institution{French Center for Scientific Research (CNRS), UMR 8188, University d'Artois}
\country{France}
}
\email{hecher@cril.fr}

\author{Matthias Lanzinger}
\orcid{0000-0002-7601-3727}
\authornote{Supported by the Vienna Science and Technology Fund (WWTF) [10.47379/ICT2201].}
\affiliation{%
  \institution{Institute for Logic and Computation, TU Wien}
  \city{Vienna}
  \country{Austria}
}
\email{matthias.lanzinger@tuwien.ac.at}

\def\defn#1{\textcolor{gray}{$\blacktriangleright$}\textbf{\textit{\boldmath #1}}}
\let\epsilon=\varepsilon

\begin{abstract}
Counting problems are fundamental across mathematics and computer science. Among the most subtle are those whose
associated decision problem is solvable in polynomial time, yet whose exact counting version appears intractable. For
some such problems, however, one can still obtain efficient randomized approximation in the form of a fully polynomial
randomized approximation scheme (FPRAS). Existing proofs of FPRAS existence are often highly technical and
problem-specific, offering limited insight into a more systematic complexity-theoretic account of approximability.

In this work, we propose a machine-based framework for establishing the existence of an FPRAS beyond previous uniform criteria. Our starting point is alternating computation: we introduce a counting model obtained by equipping alternating Turing machines with a transducer-style output mechanism, and we use it to define a corresponding counting class $\ourclass$. We show that every problem in $\ourclass$ admits an FPRAS, yielding a reusable sufficient condition that can be applied via reductions to alternating logspace, polynomial-time computation with output. We situate $\ourclass$ in the counting complexity landscape as strictly between \#L and TotP (assuming RP $\neq$ NP) and observe  interesting conceptual and technical gaps in the current machinery counting complexity. Moreover, as an illustrative application, we obtain an FPRAS for counting answers to counting the answers Dyck-constrained path queries in edge-labeled
graphs, i.e., counting the number of distinct labelings realized by $s$--$t$ walks whose label sequence is well-formed with respect to a Dyck-like language. To our knowledge, no FPRAS was previously known for this setting. We expect the alternating-transducer characterization to provide a broadly applicable tool for establishing FPRAS
existence for further counting problems.
 \end{abstract} 

\begin{document}

\maketitle

\section{Introduction}

Counting is a fundamental task in mathematics~\cite{DoubiletRotaStanley1972}, computer science~\cite{DomshlakHoffmann07a,SangBeameKautz05a,ChaviraDarwiche08a}, as well as other natural sciences~\cite{goldberg2010complexity,dirks2003partition}.
	Motivated by the crucial role of counting, there are emerging annual competitions, e.g., for propositional model counting (\#SAT)~\cite{FichteHecherHamiti21a} and path counting~\cite{InoueEtAl23}, which have applications for microchip design. 
The computational complexity of foundational counting problems has been studied for a long time. In fact, it has been known since the late 70's that counting problems can be \emph{computationally hard despite easy corresponding decision problems}~\cite{Valiant79,Valiant79b}. To name a few of these canonical problems, consider \#2CNF (sometimes denoted \#2SAT), which aims at counting satisfying assignments of a formula in conjunctive normal form with at most $2$ literals per clause, and \#NFA, asking to compute the number of accepting words up to size $n$ of a non-deterministic finite automaton.
Conceptually, both problems might seem to have similar properties and complexities, as both are contained in \#P~\cite{Valiant79}. However, while \#2CNF can not be approximated~\cite{Zuckerman96,DyerEtAl03,sly2010computational} in polynomial time\footnote{The result follows from inapproximability of counting independent sets~\cite[Theorem 2]{sly2010computational}, since independent sets can be trivially counted via \#2CNF using one clause per edge, but \#2DNF can be efficiently approximated~\cite{ArenasEtAl21}.}, unless RP = NP, \#NFA is complete~\cite{ArenasEtAl21} for the complexity class \emph{spanL}~\cite{AlvarezJenner93} and therefore admits an efficient \emph{fully polynomial-time randomized approximation scheme (FPRAS)}. The spanL-completeness result~\cite{ArenasEtAl21} has been more than celebrated, as it allows for a more mathematically rigorous model (purely based on automata) without the need of hard-to-grasp duplicate elimination via ``span'' as a postprocessing step. Note that this postprocessing step might seem minor, but it is the \emph{crucial difference between hard-to-count and easy-to-count} problems. Indeed, \#L is the class of counting problems solvable by counting all accepting paths of a non-deterministic log-space Turing machine (which is polynomial-time computable), whereas spanL eliminates duplicates (before counting) among all labeled accepting paths.

Unfortunately, the \emph{barrier between approximability and inapproximability} has been blurred for a long time. 
This blur might have been strengthened when disregarding the fine difference between Turing reductions and Karp (many-one) reductions. While for decision problems, the difference is often not that severe, for counting there is a big gap, even between a many-one reduction and one with minor postprocessing. Indeed, while \#3CNF is \#P-hard under many-one reductions, \#2CNF is not, unless P=NP, as we could use such a many-one reduction to decide 3CNFs (NP-hard) in polynomial time via an efficient algorithm for deciding 2CNFs. This argument even extends to so-called $c$-monious or affine many-one reductions\footnote{\label{foot:affine}In this work, we consider \emph{affine polynomial-time reductions} which generalizes parsimonious reductions, as in affine reductions, the number $\#_t$ of solutions of the target problem must be $a\cdot \#_s + b$, where $\#_s$ is the number of solutions of the source problem and $a,b$ are constants. If $b=0$, the reduction is called \emph{$a$-monious} and $1$-monious means \emph{parsimonious}.}. Even further, as recently shown~\cite{2SATGap}, there is a big gap between \#2DNF (in spanL) and \#2DNF plus minor arithmetic postprocessing like subtraction (which contains \#P - \#P). 
As mentioned above, it is expected that not even \#2CNF (equivalent to ``$2^n$ - \#2DNF'') can be efficiently approximated; consequently, to characterize the thin barrier between problems admitting an FPRAS and those that are among the hardest in \#P, we \emph{must not use Turing reductions} and crucially rely on the more fine-grained many-one reductions.

In many cases, FPRASes are built from scratch using new approaches and smart on-the-fly tricks during crafting. There, sometimes a new FPRAS pops up, potentially improving~\cite{MeelEtal24} or extending~\cite{MeelColnet25} the class of problems that admit an efficient FPRAS. To explicitly name an interesting example, recently it has been established that \#CFGs admit an FPRAS~\cite{MeelColnet26}, which significantly extends the known \#NFA approximability result. 
However, since \#NFA is spanL-complete we know that any problem in spanL can be efficiently approximated. This tells us that \emph{complexity classes with the right machine model} can be very useful as they can enable powerful results and meta-statements, so probably a more complexity-theoretic characterization is in order to substantially widen our understanding of the class FPRAS without the need for hand-crafted approximation schemes. In recent counting complexity works, e.g., \cite{AchilleosCalki23}, the status quo has been phrased as: ``\emph{spanL is the only counting class so far defined in terms of TMs,
that [$\dots$] contains only approximable problems}''. We pickup this challenge and address the shortcomings as follows.

\smallskip
\noindent
\textbf{Contribution 1: A New Machine Model for FPRAS.}
In this paper, we identify a unifying character of problems that admit an FPRAS that allows us to extend the class of problems that can be efficiently approximated.
We define the class \ourclass of counting problems that can be polynomial-time reduced to counting the span (removes isomorphic output) over proof trees of alternating Turing machines that use up to log space and run in polynomial time. For this class, we define the novel concept of transducers for alternating Turing machines. Technical details are carried out in Section~\ref{sec:atm}. This machine model extends
the capabilities of spanL, as it adds universal branching to the machine model, while still preserving FPRAS existence.
It turns out that \ourclass contains many natural counting problems; therefore, we provide a
useful extended machine model for easier proofs of FPRASability. 

\smallskip
\noindent\textbf{Contribution 2: New Complexity-Class Inclusions and Separations.}
Indeed, we show how \ourclass naturally extends spanL as well as very recent FPRAS results for \#CFG~\cite{MeelColnet26},
as \#CFG is contained in this class.
Beyond that, the class \ourclass of problems is FPRASable and contains spanL\footnote{The inclusion is strict under affine log-space reductions, assuming NL $\neq$ LOGCFL, see also Figure~\ref{fig:cc} and Proposition~\ref{prop:strict}.}. Since all problems in \ourclass admit an FPRAS, unless NP = RP, \ourclass must be strictly contained in \#P under affine polynomial-time reductions$^{\ref{foot:affine}}$. 
The class \ourclass is actually strictly contained in the subclass TotP of \#P. Its relationship to the most important existing complexity classes is highlighted in Figure~\ref{fig:cc} and details are established in Section~\ref{sec:spanatm}.
We thereby increase the problems that admit an FPRAS via our machine model and insights from complexity theory. %

\begin{figure}[t]
    \centering
    \includegraphics[width=0.9\linewidth]{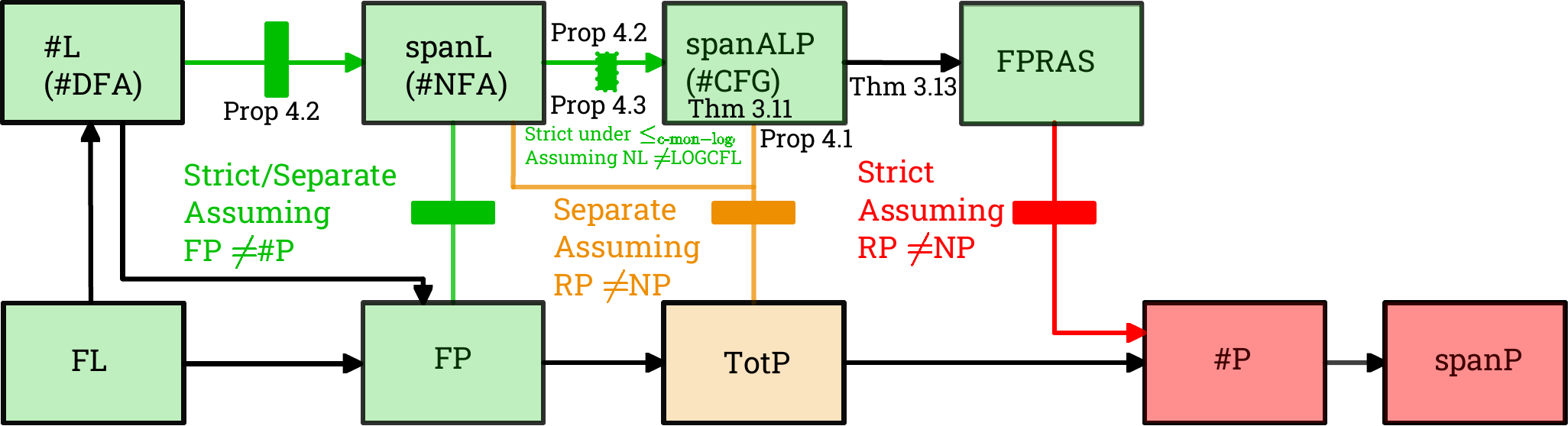}\vspace{-.65em}
    \caption{Counting complexity classes and their relation under affine polynomial-time reductions (except Proposition~\ref{prop:strict}, which uses $c$-monious log-space reductions). A directed arrow from $A$ to $B$ indicates that $A$ is included in $B$. Color coding indicates easiness based on FPRAS, however spanL still contains hard counting problems (under FP $\neq$ \#P): Green classes contain problems that can be efficiently approximated, orange classes are intermediate (easy decision problems) and red classes comprise very hard counting problems.  %
    }~\\[-1em]
    \label{fig:cc}
\end{figure}

\smallskip
\noindent\textbf{Contribution 3: A Novel Method for Obtaining an FPRAS.}
Indeed, \ourclass extends FPRAS existence beyond spanL.
Characterization via \ourclass therefore pushes the boundary of the existence of an FPRAS, enabling efficient approximations of new counting problems where an FPRAS was previously unknown. 
Among these, in Section~\ref{sec:paths} we highlight counting over Dyck-Constrained Path Queries, %
    where we only count well-formed Dyck-like walks among labeled walks in graphs.  
    This is an interesting demonstration of our class \ourclass, as counting such formatted walks
    seems to push the barrier of context-free languages.
    As an immediate consequence of our results, we immediately obtain an FPRAS for this problem. %

\smallskip
\noindent\textbf{Related Work.}
Counting problems are of reasonable interest in the database community, and there is a plethora of results in this direction. Research interests cover, for example, counting solutions to queries~\cite{ArenasEtAL22}, or, specifically, to conjunctive queries~\cite{GrecoScarcello14,ChenMengel15,ChenEtAl23}, or within conjunctive queries~\cite{KhalfiouiWijsen23}. As a corollary of our results, we provide an alternative proof of an FPRAS for counting solutions to acyclic conjunctive queries.

Recently, $[{\#2\text{CNF}}]^{\log}$, which is the class of problems that are log-space reducible to \#2CNF, has been studied~\cite{2SATGap}. Interestingly, under NL $\neq$ NP and log-space reductions, this class is strictly contained in \#P~\cite{2SATGap} and is interesting, as it is expected to be orthogonal to spanL (since \emph{\#2CNF is not expected to have an FPRAS, whereas spanL is in FPRAS}). Very recent results on \#CFG admitting an FPRAS~\cite{MeelColnet26} have been published and will be presented at SODA 2026. 
This \#CFG result motivates the search for a more general machine model. We develop an alternating transducer framework that encompasses spanL and goes beyond it, while still guaranteeing the existence of an FPRAS for every problem it captures.

There is also research on the \emph{gap of problems}. Given two \#3CNF formulas deciding their difference $[$\#3CNF - \#3CNF$]$ is complete for the class GapP \cite{FennerFortnowKurtz94}. Recent work has shown that this classification holds even if both formulas are (restricted versions of) monotone formulas in \#2CNF or even \#2DNF \cite{2SATGap}. Problem $[$\#PerfectMatching - \#PerfectMatching$]$ has been studied as well \cite{bakali2024power}, with various decision versions complete for classes such as C$_{=}$P, PP, and WPP. An exponential lower bound based on rETH was also shown. 

Closure properties of~\#P and other counting complexity classes have been extensively studied, see, e.g.,~\cite{FortnowReingold96,HoangThierauf05,OgiharaEtAl96,ThieraufEtAl94}.
There are also interesting findings on the closure of~PP under intersection~\cite{BeigelEtAl95}, which uses closure properties of \#P.

\smallskip
\noindent\textbf{Structure.}
Necessary technical preliminaries, an overview over relevant counting complexity classes, and alternating Turing machines are given in \Cref{sec:prelims}. We introduce our model of alternating transducers in \Cref{sec:atm} and discuss key complexity theoretic relationships in \Cref{sec:complexity}. We show an application to a new query language for well-formed walks in \Cref{sec:dyck}. 

\section{Preliminaries}
\label{sec:prelims}

Given a \emph{non-deterministic Turing machine} $M$, let $acc_M(x)$ be the number of accepting paths of $M$ starting from input $x$. Let $rej_M(x)$ be the same but for rejecting paths. Then define $tot_M(x) = acc_M(x) + rej_M(x)$. Note that we require every path to either accept or reject so $tot_M(x)$ is the total number of paths of the machine. Let $time(M)$ and $space(M)$ be the worst-case time and space of $M$ across all branches. If not stated otherwise, $M$ is a polynomial-time machine. We define the usual counting complexity classes.

    \defn{FL:} Counting problems that are log-space reducible to $acc_M(x)$ on deterministic log-space~$M$. 
    
    \defn{FP:} Counting problems, polynomial-time reducible to $acc_M(x)$ on deterministic poly-time~$M$. 
    
    \defn{\#L:} Counting problems that are log-space reducible to $acc_M(x)$ on log-space~$M$. 

    \defn{\#P:} Counting problems that are polynomial-time reducible to $acc_M(x)$. 

    \defn{TotP:} Counting problems that are polynomial-time reducible to $tot_M(x)$.

    \defn{Self-Reducibility:} A function $f: \Sigma^* \rightarrow \mathbb{N}$ is called polynomial-time \textbf{self-reducible} if there exist polynomials $r$ and $q$, and polynomial-time computable functions $h:\Sigma^* \times \mathbb{N} \rightarrow \Sigma^*$, $g:\Sigma^* \times \mathbb{N} \rightarrow \mathbb{N}$, and $t:\Sigma^* \rightarrow \mathbb{N}$ such that for all $x \in \Sigma^*$ the following holds (as in \cite{pagourtzis2001complexity}).
    \begin{enumerate}[(a)]
        \item $f$ can be processed recursively by reducing x to a polynomial number of instances $h(x,i)$ for $0 \leq i \leq r(|x|)$ \,i.e. formally $f(x) = t(x) + \sum_{i = 0}^{r(|x|)}g(x,i)f(h(x,i))$.
        \item the recursion terminates after polynomial depth and every instance invoked in the recursion is of poly($|x|$) size (that is, the value $f$ on instances $h(\cdots h(h(x,i_1),i_2)\cdots,i_{q(|x|)})$ can be computed deterministically in polynomial time). 
    \end{enumerate}

\noindent
A \emph{Transducer machine}  $M$ is a type of non-deterministic Turing machine that has a write-only write-once output tape. Then, each accepting path of $M$ can be assigned to its corresponding content on the output tape, but the outputs for different paths could be identical. By $span_M(x)$ we refer to the number of different outputs among all accepting paths of $M$ on input $x$.
This concept induces the following complexity classes.

 \defn{spanL:}  Counting problems that are log-space reducible to $span_M(x)$ on log-space $M$.

\defn{spanP:} Counting problems that are polynomial-time reducible to $span_M(x)$ on poly-time~$M$.

\noindent
An overview of these classes and their relationship is visualized in Figure~\ref{fig:cc}.
If there are conditional separations or conditional strict inclusions between classes $A$, $B$,
we sometimes abbreviate this by $A \not=_{c} B$ or by $A\subsetneq_{c} B$ for some condition $c$, respectively.

\begin{definition}[FPRAS]\label{def:fpras}
Let \( f : \{0,1\}^* \to \mathbb{R}_{\ge 0} \) be a function. A randomized algorithm 
\( A(x, \varepsilon, \delta) \) is called a \emph{Fully Polynomial Randomized Approximation Scheme (FPRAS)} 
for \( f \) if for every input \( x \in \{0,1\}^* \), every error parameter 
\( \varepsilon > 0 \), and every failure probability \( \delta > 0 \), the algorithm outputs a random variable 
\( Y \) such that
\(
\Pr\!\left[ (1 - \varepsilon) f(x) \le Y \le (1 + \varepsilon) f(x) \right] \ge 1 - \delta,
\)
and the running time of \( A \) is bounded by a polynomial in 
\( |x| \), \( 1/\varepsilon \), and \( \log(1/\delta) \); that is,
\(
\text{Time}(A) = \operatorname{poly}\!\left( |x|, \frac{1}{\varepsilon}, \log\frac{1}{\delta} \right).
\)
Weakening this definition to a quasipolynomial bound in $|x|$, defines the analogous notion of a \emph{Fully Quasipolynomial Randomized Approximation Scheme (QPRAS)}.
\end{definition}

By slight abuse of notation, we also use FPRAS or QPRAS as the \emph{class of problems} that admit an FPRAS or QPRAS, respectively.

\subsection{Counting-Sensitive Reductions}
The inclusions between classes of Figure~\ref{fig:cc} already demonstrate why we need a stricter reduction model. Indeed, under Turing reductions all of these classes are identical.
In this paper, we are mostly concerned with \emph{affine polynomial-time reductions} from a problem $A$ to a problem $B$. These are polynomial-time reductions~$\pi: A \rightarrow B$ that map any instance $I_A$ of $A$ to an instance of $B$ such that $\#(\pi(I_A)) = c\#(I_A) + d$ for some constants $c,d$, where $\#(I)$ is the number of solutions of $I$. If $d=0$, the reduction $\pi$ is \emph{$c$-monious} and a $1$-monious reduction is called \emph{parsimonious}. The existence of an affine polynomial-time reduction $\pi$ is denoted by $A\leq_{\mathsf{aff}\text{-}P} B$, $c$-monious polynomial-time reducibility is given by $A\leq_{c{-}\mathsf{mon}\text{-}P} B$, and if $c=1$ (parsimonious) we write $A\leq_{\mathsf{pars}\text{-}P} B$. The affine polynomial-time closure of a class $\mathcal{C}$ of problems is denoted by $[\mathcal{C}]^{\mathsf{aff}\text{-}P} = \{B \mid A\in \mathcal{C}, A \leq_{\mathsf{aff}\text{-}P} B\}$ and the closure for a problem $[{A}]^{\mathsf{aff}\text{-}P} = [\{A\}]^{\mathsf{aff}\text{-}P}$. Similarly, we denote
affine log-space reducibility ($A \leq_{\mathsf{aff}\text{-}\log} B$) and the affine log-space closure ($[\mathcal{C}]^{\mathsf{aff}\text{-}\log}$).
Unless otherwise mentioned, we assume the usage of affine polynomial-time reducibility ($\leq_{\mathsf{aff}\text{-}P}$).

A standard log-space or polynomial-time decision reducibility from $A$ to $B$ is denoted by $A\leq_{\log} B$ or $A\leq_P B$, respectively.

\subsection{Context-Free Grammars}

A \emph{context-free grammar} (CFG) is a tuple $G = (N, \Sigma, R, S)$, where $N$ is a finite set of nonterminal symbols, $\Sigma$ is a finite set of terminal symbols with $N \cap \Sigma = \emptyset$, $R \subseteq N \times (N \cup \Sigma)^{*}$ is a finite set of production rules, and $S \in N$ is the start symbol.
A \emph{derivation tree}  for $G$ is a finite rooted ordered tree $T$ whose nodes are labeled by symbols from $N \cup \Sigma \cup \{\varepsilon\}$ such that:
(i) the root is labeled by $S$, and
(ii) for every internal node $u$ labeled by a nonterminal $A\in N$, if the children of $u$ (from left to right) are labeled
$X_1,\ldots,X_m$, then $A \to X_1\cdots X_m$ is a rule in $R$ (where $m=0$ corresponds to a rule $A\to\varepsilon$).
The \emph{yield} of $T$, denoted $\operatorname{yield}(T)$, is the word obtained by reading the leaf labels of $T$ from left to right and
concatenating all symbols in $\Sigma$ (ignoring $\varepsilon$).
A word $w\in\Sigma^*$ is \emph{generated} by $G$ if there exists a derivation tree $T$ of $G$ with $\operatorname{yield}(T)=w$. The \emph{language generated by $G$} is $L(G) := \{\, \operatorname{yield}(T) \mid T \text{ is a derivation tree of } G \,\}$.
A context-free grammar $G$ is \emph{unambiguous} if every word $w \in L(G)$ has exactly one derivation tree in $G$.

A CFG is in \emph{Chomsky normal form} (CNF) if every rule in $R$ is of one of the forms $A \to BC$ with $B,C \in N$, or $A \to a$ with $a \in \Sigma$, and additionally $S \to \epsilon$ is allowed only if $\epsilon \in L(G)$ and $S$ does not appear on the right-hand side of any rule. Every CFG $G$ has an equivalent CFG $G'$ in CNF such that $L(G') {=} L(G)$ (up to presence of $\epsilon$), and $G'$ can be constructed in $O(\|G\|^{2})$ time~\cite{HopcroftUllman79}.

For $n \in \mathbb{N}$, let $L_n(G)$ denote the set of those words in $L(G)$ that have length exactly $n$. The typical counting problem for CFGs takes as input a CFG $G$, and an $n$ in unary and asks for $|L_n(G)|$. We denote this problem by \#CFG.
Recently, it has been shown that this problem admits an FPRAS.
\begin{proposition}[Theorem 1 in \cite{MeelColnet26}]\label{prop:fpras}
    \#CFG for unary size bound $n$ admits an FPRAS.
\end{proposition}

From this result, we immediately obtain the following FPRAS result.

\begin{corollary}\label{lem:cfgfpras}
    Counting the number of accepted words of a context-free grammar up to unary size bound $n$, denoted \#CFG$_{\leq}$,
    admits an FPRAS.
\end{corollary}
\begin{proof}

By Proposition~\ref{prop:fpras}, there is an FPRAS approximating \#CFG for given unary size bound $n$ via $Y_n$. 
Then, according to Definition~\ref{def:fpras},
\qquad\(
\Pr\!\left[ (1 - \varepsilon) f_n(x) \le Y_n \le (1 + \varepsilon) f_n(x) \right] \ge 1 - \delta,
\)

which immediately gives us the following probability bound for every $Y_i$ with $0\leq i \leq n$:
\(
\Pr\!\left[ (1 - \max_{0\leq j\leq n}\varepsilon_j) f_i(x) \le Y_i \le (1 + \max_{0\leq j\leq n}\varepsilon_j) f_i(x) \right] \ge 1 - \max_{0\leq j\leq n}\delta_j. \quad 
\)
Hence, we obtain: 
\(
\Pr\!\left[ (1 - \max_{0\leq j\leq n}\varepsilon_j) \Sigma_{0\leq j\leq n} f_j(x) \le \Sigma_{0\leq j\leq n} Y_j \le (1 + \max_{0\leq j\leq n}\varepsilon_j) \Sigma_{0\leq j\leq n} f_j(x) \right] \ge 1 - \max_{0\leq j\leq n}\delta_j.
\)

As a result, we can sum up the individual FPRAS values $Y_j$ for every $0\leq j\leq n$ that are obtained via Proposition~\ref{prop:fpras}, to obtain a combined FPRAS for \#CFG$_{\leq}$.
\end{proof}

\subsection{Alternating Turing Machines \& The Complexity Class LOGCFL}\label{sec:logcfl}
An \emph{alternating Turing machine} (ATM)~\cite{chandra1981alternation} is a generalization of a non-deterministic Turing machine that includes two types of states: \emph{existential} and \emph{universal}. In addition, as usual, there are also designated accepting and rejecting states.
We can view computations of ATMs as trees of configurations. An \emph{accepting computation tree} of an ATM $M$ on input string $x$ is a tree, where the root is the initial configuration of $M$ with input $x$, the descendants of universal configurations $C$ are all successor configurations of $C$, the descendants of existential configuration $C'$ is some successor configuration of $C'$, and all leaves are accepting configurations.
We say that a language $\mathcal{L}$ is accepted by an ATM  $M$ \emph{with tree-size $Z(\cdot)$} if for every string $x \in \mathcal{L}$ there is at least one accepting computation tree of $M$ on $x$ with at most $Z(|x|)$ nodes (and for $x\not \in \mathcal{L}$ there is no accepting computation tree).

The rationale for introducing tree-size bounds is that this allows us to reach an intermediate between non-determinism and full alternation. For example, a log-space ATM with polynomial tree-size bound $Z$ intuitively induces a complexity class between $\mathsf{NL}$ and an logarithmic space ATM without any explicit tree-size bound.
We write $\alt(S(n), Z(n))$ to denote the class of languages accepted by an ATM with tree-size $Z(n)$ and using $S(n)$ space, where $n$ is the length of the input.

The technical development in this paper will focus on counting via ATMs. One of the interesting properties of $\alt(\log n, \poly n)$ is that is admits various natural characterizations with alternative models of computation. We briefly introduce them here as they are relevant for later discussion. The most prominent equivalent class is $\mathsf{LOGCFL}$, the class of all problems that are log-space reducible to the membership problem in a context free language~\cite{cook1971characterizations}.
Moreover, an \emph{auxiliary pushdown automaton} (AuxPDA) is a TM with an additional pushdown stack that does count towards the space bounds on the work tape. Time is measured as usual~\cite{cook1971characterizations}.  Let $\mathsf{NAuxPDA}(S(n),Z(n))$ be the languages decidable by a nondeterministic AuxPDA in work space $O(S(n))$ and time $O(Z(n))$. 
Finally, $\mathsf{SAC}^1$ is the class of problems solvable by $\mathsf{AC}^0$-uniform, polynomial-size circuits of depth $O(\log n)$ in
which $\vee$-gates have unbounded fan-in and $\wedge$-gates have fan-in $2$ (with negations only at inputs), see, e.g.,~\cite{Immerman99}.
It is then known~\cite{sudborough1978tape,ruzzo1979tree,cook1971characterizations,venkateswaran1987properties} that:
\[
\mathsf{LOGCFL}=\mathsf{ALT}(O(\log n),\poly(n))=\mathsf{NAuxPDA}(O(\log n),\poly(n)) = \mathsf{SAC}^1.
\]

To avoid notational clutter, below we will write $\log,\poly$ instead of $O(\log n),\poly(n)$ for the tree size and space parameters of ATM complexity classes.

\section{Counting via Alternating Turing Machines}\label{sec:atm}

Let $M$ be an ATM. Recall that acceptance of $M$ on input string $x$ is defined via the existence of a proof tree. Proof trees depend on non-deterministic choices in the existential states, hence there can be multiple accepting proof trees. Since we care about proof tree size, we have to consider only proof trees adhering to bounds (recall that even if the decision problem can be solved with $Z(n)$ proof trees, it does not mean that all proof trees in the ATM are of that size, only that at least one accepting on is). For $Z : \mathbb{N} \to \mathbb{N}$ define $\mathsf{AT}_{Z}(M,x)$ be the set of all accepting computation trees of $M$ for input $x$ of size at most $Z(|x|)$ and using at most $|x|$ space. Let $\#_Z(M,x) = |\mathsf{AT}_Z(M,x)|$ be the number of accepting computation trees. This leads to a natural ATM-based counting complexity~class.

\begin{definition}
    Let $\#\alt(S(n), Z(n))$ be the class of counting problems $f : \Sigma^* \to \mathbb{N}$ such that there is an ATM $M$ using at most $S(n)$ space and where $f(x) = \#_Z(M,x)$ on all $x\in \Sigma^*$.
\end{definition}

This gives us a natural counting class that corresponds cleanly to LOGCFL, namely the class $\#\alt(\log n, \poly(n))$.
While this definition of \#\alt is straightforward, it is insufficient to capture more interesting problems, cf. \#L vs.\ \spn{}L.
To define a class that is strong enough, we therefore introduce machinery that allows us to reason about duplicate witnesses in ATMs.  

\subsection{Alternating Transducer Machines and \spn{}\alt}

We now extend alternating Turing machines to \emph{output-producing} devices, in order to define a natural notion of \emph{span} counting under alternation.

\begin{definition}[Alternating Transducer Machine]
An \emph{Alternating Transducer Machine (ATrM)} is an ATM
\(
M = (Q, \Sigma, \Gamma, q_0, \delta, t, \prec)
\)
augmented with a write-only, write-once output tape, such that:
\noindent\begin{inparaenum}%
    \item $Q$ is a finite set of states;
    \item $\Sigma$ is the input alphabet;
    \item $\Gamma$ is the work-tape alphabet (contains $\Sigma$);
    \item $q_0 \in Q$ is the initial state;
    \item $t : Q \to \{\exists, \forall, \mathsf{accept}, \mathsf{reject}\}$ labels each state by its type; and\\
    \item $\delta$ is a \emph{transition function}
    \(
        \delta : Q \times \Gamma \longrightarrow
        \mathcal{P}\bigl(Q \times \Gamma \times \{L,R\} \times \Sigma_\circ\bigr),
    \)
    where $\Sigma_\circ := \Sigma \cup \{\circ\}$ and $\circ$ denotes the ``no-output'' symbol.
\end{inparaenum}
A transition $(q',\gamma',D,\omega)$ in $\delta(q,\gamma)$ has the usual meaning for the control, work tape, and head movement. %
Moreover, $\prec$ is a strict linear order on $Q \times \Gamma \times \Sigma_\circ$, which serves to order the children of universal states.
\end{definition}

Through the output symbols in the transition function we will then define the output of an accepting computation tree. \emph{Some care is required here}, as the option to not create any output on universal states can create situations where it is unclear what a natural interpretation of the output would be. Consider the following two accepting computation trees, with configurations labeled by the output symbol of the transition that moves into the configuration.
\begin{figure}[H]~\\[-1em]
\centering
\begin{subfigure}[b]{0.2\linewidth}
\centering
\begin{tikzpicture}[baseline=(root.base),
  level distance=4mm,
  sibling distance=10mm,
  every node/.style={inner sep=1pt}]
\node (root) {$\circ$}
  child { node {$a$}
    child { node {$\circ$}
      child { node {$c$} }
      child { node {$d$} } }
    child { node {$b$} } }
  child { node {$e$}};
\end{tikzpicture}
\caption{}
\end{subfigure}
\hspace{1cm}
\begin{subfigure}[b]{0.2\linewidth}
\centering
\begin{tikzpicture}[baseline=(root.base),
  level distance=4mm,
  level 1/.style={sibling distance=12mm},
  level 2/.style={sibling distance=10mm},
  every node/.style={inner sep=1pt}]
\node (root) {$\circ$}
  child { node {$a$}
    child { node {$b$} }
    child { node {$c$} } };
\end{tikzpicture}
\caption{}
\end{subfigure}
\hspace{1cm}
\begin{subfigure}[b]{0.2\linewidth}
\centering
\begin{tikzpicture}[baseline=(root.base),
  level distance=4mm,
  level 1/.style={sibling distance=12mm},
  level 2/.style={sibling distance=10mm},
  every node/.style={inner sep=1pt}]
\node (root) {$\circ$}
  child { node {$a$}
    child { node {$b$} } }
  child { node {$a$}
    child { node {$c$} } };
\end{tikzpicture}
\caption{}
\end{subfigure}~\\[-1em]
\end{figure}
Tree (a) illustrates why defining the output as a \emph{single labeled tree is delicate}, once we allow $\circ$-transitions. If one simply deletes $\circ$-labeled nodes, then the root disappears and the remaining labeled nodes form a \emph{forest} (here, an $a$-component and an $e$-leaf), rather than a rooted tree. Moreover, deleting an internal $\circ$-node forces an arbitrary choice between keeping an unlabeled branching node or flattening its children into the nearest labeled ancestor. On the other hand, keeping $\circ$ nodes is problematic as it would require precise analysis of the number of computation steps (since they all create a node) between outputs to argue whether two outputs are identical.

Trees (b) and (c) illustrate a problem with another natural approach: writing to an output tape. The natural version of writing to an output tape would 
create an individual copy of the tape for each child in a universal state. It is then \emph{tempting to consider the states of the output tapes at the leaves} of the computation trees as the output of the tree. However, in (c), both the left and right subtree look identical at their leaves. Both contain the string $ab$ and $ac$, respectively. Yet, it seems desirable to consider computation trees (b) and (c) as different.

\smallskip
We resolve these issues with the following definition of output for ATrMs.

\begin{definition}[ATrM Output]
We define an \emph{accepting computation tree} of an ATrM analogous to accepting computation trees of ATMs. Additionally, we add the label $\sigma(m)$ for each node $m$ in the tree, which maps the node to the output symbol of the transition that reached the corresponding configuration. For the root node $r$, $\sigma(r)$ maps to $\circ$.
The \emph{output} of an accepting computation tree $T$, written
$\mathsf{Out}(T)$, is defined by collecting the outputs into a tuple of trees as specified in \Cref{alg:outT}.
\end{definition}

\begin{algorithm}[t]\setstretch{0.85}
\small
\caption{\textsc{Out}$(T)$: output of an accepting computation tree $T$}
\label{alg:outT}
\DontPrintSemicolon
\SetKwFunction{Out}{OutNode}
\SetKwProg{Fn}{Function}{:}{}
\Fn{\Out{$u$}}{
  $\mathcal{C}\gets \langle\,\rangle$\;
  \ForEach{child $v$ of $u$ in $\prec$ order}{
      concatenate $\Out(v)$ to the end of $\mathcal{C}$\;
  }
  \lIf{$\sigma(u)=\circ$}{
      \KwRet{$\mathcal{C}$}
  }
  $\tau \gets$ an ordered tree with only node $r$ with label $\sigma(u)$\;
  \ForEach{$S\in\mathcal{C}$ in  $\prec$  order}{
      add $S$ to $\tau$ with the root of $S$ becoming a child of $r$\;
  }
  \KwRet{$\langle \,\tau\,\rangle$}\;
}
\BlankLine
\KwRet{\Out{root($T$)}}\;
\end{algorithm}

\begin{definition}[Output Set and Span]
For $S, Z : \mathbb{N} \to \mathbb{N}$, let  
$\mathsf{ACT}_M(x; S, Z)$
denote the set of all $\mathsf{Out}(T)$ of accepting computation trees $T$ of $M$ on input $x$, where $T$ uses at most $S(|x|)$ space and has at most $Z(|x|)$ nodes.
We define the \emph{span} of $M$ as
\[
    \spn_M(x; S, Z)
      := \bigl|\, \{\mathsf{Out}(T) \mid T \in \mathsf{ACT}_M(x; S, Z)\}  \, \bigr|, i.e.,
\]
the number of distinct (non-isomorphic) outputs of accepting computation trees within given~bounds.
\end{definition}

\begin{definition}[The Class $\spn\alt(S,Z)$]
\label{def:spanalt}
The class $\spn\alt(S,Z)$ consists of all functions
$f : \Sigma^* \to \mathbb{N}$
for which there exists an ATrM $M$ such that
\(
    f(x) = \spn_M(x; S, Z)
\)
for all $x \in \Sigma^*$.
\end{definition}

\Cref{def:spanalt} is now the ATM analogue of the established span complexity classes (spanL, spanP) that we originally wanted. Inspired by the recent FPRAS result for \#ACQ and the LOGCFL-completeness of the corresponding decision problem, we also focus on the class $\spn{}\alt(\log, \poly)$ as a natural counting version of LOGCFL. However, since we are primarily interested in finding the largest possible complexity class for counting problems such that problems in the class are guaranteed to have an FPRAS, we will slightly expand this definition to define the class \ourclass.

The natural question of whether $\spn\alt(\log, \poly)$ is already closed under polynomial time many-one reductions is difficult to approach. Although interesting, this question is not critical in terms of our goal of determining the existence of efficient approximations (FPRAS), but we will argue about conditional non-closedness in Section~\ref{sec:spanatm}.

\medskip
We conclude our development of ATrMs with the following technical observation that allows us to assume that our ATrMs behave simply, universal branches are always binary, and without loss of generality, we can assume that there are no accepting computation trees that are beyond the tree size bound we consider.

\begin{proposition}
\label{prop:normal.atrm}
    Let $P$ be a counting problem in $\spn\alt(S,Z)$ where $Z \in \poly(|x|)$. Then there exists $S' \in O(S+\log|x|)$, $Z' \in \poly(|x|)$, and an ATrM $M$ with the following properties:
    \begin{itemize}
        \item Every universal state $q$ of $M$ has $|\delta(q,\gamma)|=2$ for every $\gamma\in\Gamma$. In other words, universal states produce exactly two computation branches.
        \item For every input $x$ we have:
        \qquad\qquad\(
        \spn_M(x; S', Z') = \spn_M(x; S', \infty) = \spn_P(x; S, Z).
        \)
    \end{itemize}
\end{proposition}
\begin{proof}[Proof Sketch]
    If a universal state $q$ has one transition, just make it existential. If it has none, make it accept; if it has more than two, branch to the least of them, and to a new universal state with the rest of them. Then iterate, which blows up the tree size by a factor of 2 at most.

    For the other property, compute $Z(|x|)$ in $\log(Z(|x|))$ space and guess a number at most $Z(|x|)$ for how many original states there will be in the accepting computation tree. Then add auxiliary states to count the number of original states passed. At universal branches, guess how to split the remaining size. At accepting states, accept iff the counter is 0, otherwise reject. Clearly, only a polynomial number of extra states and space $\log(Z(|x|)) = \log(|x|)$ (recall $Z {\,\in\,} \poly(|x|)$) is~needed.
\end{proof}

\begin{algorithm}[t]\setstretch{0.85}
\small
\caption{\textsc{ACQTransduce}: alternating transducer for \#ACQ given a join tree}
\label{alg:acq}
\textbf{Input:} Acyclic CQ $q(\bar{x})=\exists \bar{y}\ \bigwedge R_n(\bar{v}_n)$, database $D$, rooted join tree $T$ of $q$ with root $r$.\\
\SetKwProg{Fn}{Function}{:}{}

\DontPrintSemicolon

\SetKwFunction{EvalNode}{EvalNode}

\Fn{\EvalNode{$n,\theta$}}{
  Let the label of $n$ be the atom $R(\bar{v})$ with $\bar{v}=(v_1,\ldots,v_k)$\;
  $\bar{a} \leftarrow$ \textbf{Guess} a tuple $\bar{a}\in R^D$\;
  \lIf{there exists $v_i \in \mathrm{dom}(\theta)$ with $\theta(v_i)\neq a_i$}{\textbf{Reject}}
  $\bar{b} \leftarrow (b_1,\ldots,b_k)$ where $b_i \gets a_i$ if $v_i\in \bar{x}$, and $b_i \gets \star$ otherwise\;
  \textbf{emit} $\bar{b}$\;
  $\theta_n \gets \{\, v_i \mapsto a_i \mid i\in[k] \,\}$\;
   
  \ForEach(\tcp*[f]{ $\forall$-branch over children }){child node $c$ of $n$ }{
    \lIf{\EvalNode{$c,\theta_n$} \textbf{Rejects}}{\textbf{Reject}} 
  }
  \textbf{Accept}\;
}
\Return \EvalNode{$r,\emptyset$}\;
\end{algorithm}
 
\begin{example}[Acyclic Conjunctive Queries]
\label{ex:cq.atrm}
We assume that the reader is familiar with conjunctive queries, and their $(\alpha$-)acyclicity (see, e.g.,~\cite{DBLP:journals/jacm/Fagin83}). The problem \#ACQ refers to the associated combined complexity counting problem, for a given acyclic CQ $q$ and database $D$, compute $|q(D)|$. 

We show how \#ACQ can be expressed as $\spn{}\mathsf{ALT}(\log, \poly)$ on ATrMs. This is analogous to the reduction of counting answers to ACQs to tree automata, given by Arenas et al.~\cite{ArenasEtAl21}, to originally show that \#ACQ admits an FPRAS. We present it here as an instructive example of observing existence of an FPRAS via our framework.

Suppose that our input is $q,D, T$,  where $T$ is a given join tree of $q$\footnote{We will see at the end of the next subsection, why assuming the join tree as input is unproblematic.}, for the ATrM described in \Cref{alg:acq}.
Each accepting computation tree corresponds exactly to a consistent assignment. The output corresponds directly to the join tree with database constants in place of output variables, and a ``$\star$'' in place of all other variables. Thus, any two variable assignments that agree on the free variables of $q$, produce the same output. It is clear that this ATrM only requires logarithmic space and that every accepting computation is polynomial size in the input.
Of particular note is that this algorithm lines up with the standard LOGCFL algorithm for the corresponding decision problem, whether an acyclic conjunctive query has a non-empty output. This translates just as well to other LOGCFL query problems. For instance, the analogous result for acyclic CRPQs follows by replacing Line~4 in \Cref{alg:acq} with a guess of 2 constants and an $\mathsf{NL}$ check for their reachability under some regular path pattern\footnote{In combined complexity, \#CRPQ $\le_{\mathsf{pars}\text{-}P} $ %
\#CQ; hence \#Acyclic-CRPQ admits an FPRAS (already via the FPRAS for \#ACQ).}. Figure~\ref{fig:acq} depicts the outputs for a particular ACQ $q(x,w)$.

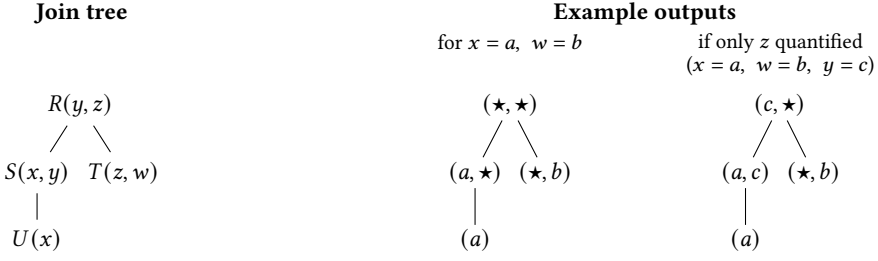
\begin{figure}[t]
\begin{tikzpicture}[every node/.style={font=\small}]

\node[, rounded corners, inner sep=2pt, align=center] (jt) at (-3.7,1.9) {\bfseries Join tree};
\node at (-3.7,-0.2) {
\begin{forest}
for tree={rounded corners, s sep=1mm, l sep=1mm, inner sep=2pt, align=center}
[{$R(y,z)$}
  [{$S(x,y)$} [$U(x)$]]
  [{$T(z,w)$}]
]
\end{forest}
};

\node[, rounded corners, inner sep=2pt, align=center] (ex) at (3.8,1.9) {\bfseries Example outputs};

\node at (2.0,1.5) {\footnotesize for $x=a,\;w=b$};
\node at (2.0,-0.2) {
\begin{forest}
for tree={rounded corners, s sep=1mm, l sep=1mm, inner sep=1.5pt}
[{$(\star,\star)$}
  [{$(a,\star)$} [{$(a)$}]]
  [{$(\star,b)$}]
]
\end{forest}
};

\node at (5.6,1.5) {\footnotesize if only $z$ quantified};
\node at (5.6,1.2) {\footnotesize ($x=a,\;w=b,\;y=c$)};
\node at (5.6,-0.2) {
\begin{forest}
for tree={rounded corners, s sep=1mm, l sep=1mm, inner sep=1.5pt}
[{$(c,\star)$}
  [{$(a,c)$} [{$(a)$}]]
  [{$(\star,b)$}]
]
\end{forest}
};
\end{tikzpicture}~\\[-1.35em]
\caption{Illustration of the outputs in \Cref{ex:cq.atrm} for $q(x,w)=\exists y\,\exists z\;\big(S(x,y)\wedge R(y,z)\wedge T(z,w) \wedge U(x)\big)$.}\label{fig:acq}
\end{figure}

\end{example}

\subsection{Relating $\spn\alt$ and \#CFG: The New \ourclass Class}

This section formally captures key relationships between the counting problems over ATrMs and \#CFG. We will also consider analogous problems \#CNFG and \#UCFG, that are defined like \#CFG but restricted to grammars in Chomsky normal form and unambiguous CFGs, respectively.

\begin{lemma}\label{lem:cnfg2alt}
     \#CNFG $\le_{\mathsf{pars}\text{-}L} \spn\alt(\log , \poly)$.
\end{lemma}~\\[-2.5em]
\begin{proof}
We directly implement \#CNFG, where the inputs are CFG $G$ in Chomsky normal form and integer $k$ (in unary), with a fixed log-space, polynomial tree-size ATrM. The respective reduction is thus simply identity.
In particular, \Cref{alg:atm-cfg-transducer} shows a $\spn\alt(\log , \poly)$ algorithm for \#CFG. For every word $w_1w_2\cdots w_k \in \mathcal{L}^k(G)$, there is an accepting computation tree, mirroring a derivation tree in $G$ of the word, that emits non-$\circ$ output only at exactly $k$ leaves. Therefore, $\mathsf{Out}(T)$ is a tuple of $k$ singleton trees. It is easy to verify that the $i$th tree in $\mathsf{Out}(T)$ has precisely $w_i$ as label. Furthermore, there are no other accepting computations.
Again, similarly to \Cref{ex:cq.atrm}, this is analogous to standard LOGCFL algorithms for problems on context-free languages.%
\end{proof}

The restriction to Chomsky normal form is critical to the reduction above. In particular, since accepting computation trees in \Cref{alg:atm-cfg-transducer} directly follow the derivation trees of a word, we need to avoid superpolynomial derivation trees that produce the empty word $\epsilon$, to remain within the desired polynomial tree size bound. In the context of our search for counting complexity classes that imply an FPRAS this is also no further cause for concern, as polynomial time pre-processing to transform a CFG into CNF does not affect the existence of an FPRAS. 
Since the reduction in \Cref{alg:atm-cfg-transducer} produces an accepting computation tree for each derivation tree of the grammar, we see that each word has a single accepting computation tree when the grammar is unambiguous, making the additional power of $\spn$ unnecessary. 

\begin{corollary}\label{cor:ucfg}
\#UCFG $\le_{\mathsf{pars}\text{-}P} \#\alt(\log , \poly)$.
\end{corollary}%
\medskip\begin{algorithm}
[b]\setstretch{0.85}
\small
\caption{Alternating transducer for length-constrained derivability in a CNF grammar}
\label{alg:atm-cfg-transducer}
\DontPrintSemicolon
\SetKw{KwAccept}{Accept}
\SetKw{KwReject}{Reject}
\SetKw{True}{true}
\SetKw{False}{false}
\SetKw{Emit}{emit}
\SetKwFunction{TransduceLen}{TransduceLen}
\SetKwProg{Fn}{Function}{:}{end}

\KwIn{A CFG $G=(N,\Sigma,R,S)$ in Chomsky normal form and integer $k\ge0$.%
}

\medskip
\Fn{\TransduceLen{$X,k$}}{
    \lIf*{$k=0$}{
    \lIf*{$X = S$ \textbf{and} $S \to \varepsilon$}{\KwAccept}\lElse{\KwReject}}
    $r \leftarrow$ Guess a rule $r\in R$ with left-hand side $X$.\;
    \lIf{no such rule $r$ exists}{ \KwReject}
    
    \uIf{$r$ is $X\to a$ for some $a\in\Sigma$}{
        \lIf*{$k=1$}{
        \Emit $a$;\text{ }\KwAccept
        }
        \lElse{\KwReject}
    }
    \uElseIf{$r$ is $X\to BC$ with $B,C\in N$}{
        $k_1\leftarrow$ Guess $k_1\in\{0,\dots,k\}$%
            \tcp*[l]{$\forall$ branch over both nonterminals}
            \lIf{\TransduceLen{$B,k_1$} \textbf{and} \TransduceLen{$C,k{-}k_1$} both \KwAccept}{ \KwAccept}
            \lElse{
             \KwReject
             }
    }
}
\BlankLine
\Return \TransduceLen{$S,k$}
\end{algorithm}

To relate alternating transducers back to CFGs, we show that the span of an ATrM can be represented as the language of an efficiently constructed grammar.

~\\[-2.5em]
\begin{lemma}
\label{lem:atrm2cfg}
    Let $M$ be an ATrM, $x$ an input string, $S,Z$ be space and tree-size bounds on $M$.
    Then, there exists a CFG $G$ such that
    \(
\spn_M(x; S,Z) = |\mathcal{L}_{\le 3Z(|x|)}(G)| \qquad \text{and} \qquad \|G\| \in O(2^{S(|x|)}).
\)

\noindent Moreover, $G$ can be computed in $O({S(|x|)})$ space.
\end{lemma}
\begin{proof}[Proof Sketch]
For $M$ with alphabet $\Sigma \cup \{\circ\}$ and input $x$, let $\mathcal{C}$ be the set of all possible configurations. Define CFG $G=(\mathcal{C}, \Sigma \cup \{\langle,\, \rangle\}, R, C_0)$ as follows:
$C_0$ is the initial configuration of $M$ with input $x$. For the brevity of rules, we first define the function $\mathsf{outwrap}$ that takes a configuration and $\sigma \in \Sigma\cup \{\circ\}$, and returns a string in $(N \cup \Sigma \cup \{\langle, \rangle\})^*$.
\[
\mathsf{outwrap}(C,\sigma) = \begin{cases}
    C & \text{if $\sigma=\circ$} \\
    \langle\sigma C\rangle & \text{otherwise}
\end{cases}
\]

The transitions of $M$ translate into following rules $R$ (reject transitions induce no rules):
\begin{align}
    & C \to \mathsf{outwrap}(C',\sigma) & \text{for all $C$ at existential states with  $(C',\sigma) \in \delta(C)$}  \\
    & C \to \mathsf{outwrap}(C_1,\sigma_1)\mathsf{outwrap}(C_2,\sigma_2) & \text{for all $C$ at universal states with}  \notag\\[-.5em]
    & & \text{$\{(C_1,\sigma_1), (C_2,\sigma_2)\} = \delta(C)$ and $C_1 \prec C_2$}\\
    & C \to \varepsilon & \text{if $C$ is an \textsf{Accept} state}
\end{align}
There is a bijection $\iota$ between words in $\mathcal{L}(G)$ and accepting computation trees of $M$ on input $x$. In particular, $\spn_M(x) = |\mathcal{L}(G)|$.
Suppose we have space and size bounds $S,Z$, respectively. Then $|\mathcal{C}|$, and by extension, $|R|$ is in $O(2^{S(|x|)})$. Moreover, from the construction of $R$, an accepting computation tree with $n \leq Z$  non-empty outputs corresponds to a word in $G$ with length $3n$. We thus have 
~\\[-2em]\[
\spn_M(x; S,Z) = |\mathcal{L}_{\le 3Z(|x|)}(G)| \qquad \text{and} \quad \|G\| \in O(2^{S(|x|)}).\\[-1.25em]
\]%
\end{proof}

\begin{theorem}[Counting CFG Characterization]
\label{thm:cfg.eq.atm}~\\[-1.5em]
\begin{align*}
    & [\spn\alt(\log,\poly)]^{\mathsf{aff}\text{-}P} &&= && [\#\mathsf{CFG}]^{\mathsf{aff}\text{-}P}\text{ and}
    \\[-.45em]
    &
    [\#\alt(\log,\poly)]^{\mathsf{aff}\text{-}P} &&= && [\#\mathsf{UCFG}]^{\mathsf{aff}\text{-}P}.
\end{align*}%
\end{theorem}
\begin{proof}[Proof Sketch]
Immediate from Lemmas~\ref{lem:cnfg2alt} and~\ref{lem:atrm2cfg} as well as the well-known fact that it requires quadratic time to compute the Chomsky normal for of a CFG $G$. %
\end{proof}

Recall, we are particularly interested in affine polynomial time reductions, as the existence of an FPRAS is preserved under them. The established terminology of complexity theory is slightly confusing in the counting context. While $\alt(\log, \poly)$ is contained in polynomial time in a decision context, a polynomial time reduction here corresponds to the complexity class \textsf{FP}, which we show below to be significantly weaker than the alternating counting classes discussed here.
In light of our overall motivation of finding natural complexity classes that imply the existence of an FPRAS, \Cref{thm:cfg.eq.atm} motivates the definition of the following class.

\begin{definition}[The \ourclass Class]
    The class \ourclass of counting problems is the closure of $\spn\alt(\log, \poly)$ under affine poly time reductions, i.e., \ourclass := $[\spn\alt(\log, \poly)]^{\mathsf{aff}\text{-}P}$.
\end{definition}
The definition might prompt the inquisitive reader to wonder whether $\spn\alt(\log,\poly)$ is closed under polynomial time reductions. Already under $c$-monious reductions, $\spn\alt(\log,\poly)$ must not be polynomial-time closed,
assuming P $\neq$ LOGCFL, which is widely expected. Indeed, if $\spn\alt(\log,\poly)$ were polynomial-time closed under $c$-monious reductions,
by the reduction in the proof of Lemma~\ref{lem:cnfg2alt} and since $c$-monious reductions preserve the outcome of underlying decision problems, 
we could decide any P-hard problem in LOGCFL (as LOGCFL $=\alt(\log,\poly)$).
Since the counting versions of some P-hard problems have no FPRAS (cf., \Cref{prop:totp}), we thus inevitably have to consider the affine polynomial-time closure over a machine model that is not inherently closed under them.

Our study of alternating transducers  culminates in the following statement that relates them to the existence of FPRAS, and QPRAS, approximations for problems within natural ATrM classes.
\begin{theorem}[FPRAS and QPRAS]\label{thm:qpras}
    $$\spn\alt(\log^c, \poly) \in \mathsf{QPRAS}, \qquad \spn\alt(\log, \poly) \in \mathsf{FPRAS}, \qquad \ourclass \in \mathsf{FPRAS}.$$
\end{theorem}
\begin{proof}
    By \Cref{lem:atrm2cfg}, we can parsimoniously reduce from a problem $P$ in $\spn\alt(\log^c,\poly)$ to \#CFG$_\le$ in quasipolynomial time. The reduction produces a quasipolynomial-sized CFG $G$. Applying the FPRAS for \#CFG~\cite{MeelColnet26} gives us a QPRAS for $P$ (which is an FPRAS for $P$ if $c=1$).
\end{proof}

This then also fills in the missing gap in~\Cref{ex:cq.atrm}, where we assumed  that a join tree was given in the input. As the join tree can be computed in LOGCFL $\subseteq \mathsf{P}$, we see that \#ACQ has an $\mathsf{aff}\text{-}P$ reduction to the problem in the example.

\paragraph*{Counting with Other LOGCFL Machine Models?}
Our development of counting complexity classes based on ATMs is partially motivated by the recent FPRAS for \#ACQ of Arenas et al.~\cite{ArenasEtAl21}. The corresponding decision problem is LOGCFL-complete and thus naturally aligns with alternation (cf.,~\Cref{sec:logcfl}). However, our study of these problems indicates that this type of linkage to the decision level becomes more subtle in a counting context. Given the alternative characterizations of LOGCFL discussed in \Cref{sec:logcfl}, it is tempting to expect matching counting classes, but so far we have not been able to establish such equivalences.

For instance, it is straightforward to define a transducer version of $\mathsf{NAuxPDA}(\log,\poly)$, and $\spn\alt(\log,\poly)$ parsimoniously reduces to the corresponding span class (see \Cref{app:nauxpda}). In the other direction, however, known reductions for the decision case are not counting preserving, and there is no obvious way to repair them. The circuit characterization poses further difficulties: it is unclear how to capture the mechanics of ``\spn{}''-type counting classes in circuits on a conceptual level. Interestingly, the LOGCFL-hardness of ACQs is proved by reducing from $\mathsf{SAC}^1$, which lines up with the discussion in \Cref{sec:complexity} suggesting that $\ourclass$ strictly subsumes $\spn{}L$, while \#ACQ lies in $[\spn{}L]^{\mathsf{aff}\text{-}P}$ by a reduction to \#NFA that counts accepting inputs of tree~automata~\cite{ArenasEtAL22}.

\section{Class \ourclass in the Counting Complexity Landscape}\label{sec:spanatm}
\label{sec:complexity}

We proceed to study where \ourclass fits within the current counting complexity landscape. 
Below, we assume affine polynomial-time reductions, unless mentioned otherwise.
First, we establish that \ourclass is strictly contained in TotP and \#P under standard assumptions.

\begin{proposition}[Strictly Included Above]
\label{prop:totp}
    \ourclass $\subsetneq_{\text{RP $\neq$ NP}}$ TotP.
\end{proposition}
\begin{proof}
    By~\cite[Theorem 3]{bakali2020characterizations} we know that \#2CNF is in TotP.
    However, unless RP=NP, \#2CNF does not have an FPRAS~\cite[Theorem 2]{sly2010computational}. However, \ourclass is in 
    FPRAS by Theorem~\ref{thm:qpras}. 
\end{proof}

Therefore, since FPRAS $\subsetneq$ \#P, we conclude that \ourclass $\subsetneq_{\text{RP $\neq$ NP}}$ \#P.
Further, since \#L is strictly contained in spanL, \#L is also strictly contained in \ourclass. %

\begin{proposition}[Inclusion of Below]
    \#L $\subsetneq_{\text{FP $\neq$ \#P}}$ spanL $\subseteq$ \ourclass.
\end{proposition}
\begin{proof}
    There is a parsimonious reduction from \#P to \#3CNF~\cite[Lemma 3.2]{Valiant79}, which is equivalent to $2^n - $\#3DNF where $n$ is the number of variables of the formula.
    Since \#L $\subseteq$ FP~\cite{AlvarezJenner93} via simulation of the log-space Turing machine, we know that if \#3DNF (contained in spanL~\cite{AlvarezJenner93}) were in \#L, we could reduce \#3CNF to FP. This would then establish that \#P = FP.
    The relationship between spanL and $\ourclass$ follows, since our machine model is an extension of the machine model of spanL.
\end{proof}

We also expect that \ourclass is a strict superset of spanL, which seems to be in line with the Chomsky hierarchy.
However, a formal separation is more challenging than a fallback to known decision hierarchies, as a potential separation also involves non-uniform computation and polynomial time. Still, we can show the following separation
under log-space reductions.

\begin{proposition}[Strict Inclusion]\label{prop:strict}
    spanL $\subsetneq_{\text{NL}\neq\text{LOGCFL}}$ $\#\mathsf{ALT}(\log, \poly) \subseteq \ourclass$ under $c$-monious log-space reductions.
\end{proposition}%
\begin{proof}
Suppose for a contradiction that we could reduce from \#ACQ$^\star$ (the variant of \#ACQ in $\#\mathsf{ALT}(\log, \poly)$, solved via Algorithm~\ref{alg:acq} without ``\spn{}'') in log-space to spanL.
Then, by construction and since the reduction is $c$-monious we could decide ACQ by means of a log-space Turing machine, i.e., in NL. This contradicts LOGCFL-completeness of ACQ~\cite{GottlobLeoneScarcello01}, assuming %
that NL $\neq$ LOGCFL.
\end{proof}

Although this argument needs log-space reductions, due to known expressiveness results, we propose the following conjecture, lifting the separation from \Cref{prop:strict} to \textsf{aff-P} reduction.

\begin{conjecture}\label{conj:spanlweaker}
    \spn{}L $\subsetneq$ \ourclass.
\end{conjecture}%
\begin{proof}[Justification]
The conjecture is motivated from two fronts. On the one hand, on the machine model level $\spn\alt(\log,\poly)$ is $\spn$L extended by universal branching. Refuting the conjecture would mean that this universal branching can be simulated by polynomial time preprocessing (via the reduction) and existential states.

On the other hand, by \Cref{thm:cfg.eq.atm} $\spn{}\text{L}=\ourclass$ would mean that \#CFG could be reduced to \spn{}L, and thus also to \#NFA~\cite{ArenasEtAl21}, in polynomial time.
\end{proof}

We note that it is known that given a CFG $G$, there cannot be an NFA $A$ such that $L_n(G) = L_n(A)$ for all $n$. 
The generating functions $\sum_k L_k(A) x^k$ for any NFA are known to be rational. However, there are context-free languages where the respective generating function is not rational and hence its coefficients cannot match the one induced by an NFA~(see \cite{chomsky1959algebraic}). However, our setting is more intricate as the word-length $n$ is in the input. Thus, a reduction from \#CFG to \#NFA would allow mapping to a different regular language for each $n$. While this setting is more challenging to analyze formally, it is still expected that this is not possible in a polynomial time reduction.

Answering Conjecture~\ref{conj:spanlweaker} is indeed linked to an open problem, as it would give a strict separation of spanL from FPRAS, which has been open for decades. Although a separation is expected, \#Permanent is not expected to fit in spanL~\cite{Valiant79b}, despite there being a known FPRAS for it~\cite{JerrumEtAl04}. However, \textbf{\ourclass is the first counting class based on extending the Turing machine model of spanL}, as also noted in the literature.
Indeed, in \cite{AchilleosCalki23} it is emphasized that ``\emph{spanL is the only counting class so far defined in terms of TMs,
that} [$\dots$] \emph{contains only approximable problems}''.

\section{Dyck-Constrained Path Queries and How to Count Them}
\label{sec:paths}\label{sec:dyck}
In this section, we demonstrate how our characterization of \ourclass allows us to 
reach beyond the state-of-the-art by establishing an FPRAS for a natural fundamental problem scheme. 
While \ourclass is equivalent to \#CFG under polynomial time reductions, we believe that the move to a Turing machine model greatly lowers the conceptual barrier for recognizing that a problem falls within this class. We illustrate this by showing that a family of natural complex graph languages admits an FPRAS via implementation in terms of ATrMs.

We consider a graph query language for querying walks that require well-formed opening/closing behavior along their paths.
Such queries capture various kinds of complex graph analysis observed in applications such as program verification~\cite{1702388}, the analysis of 
control flow graphs~\cite{Muchnick1997}, and physics~\cite{SCHRAM2013891}.
Indeed, our well-formed s-t walks are extensions of labeled s-t walks,
for which it was unknown whether these can be counted via an FPRAS.
Classical approaches seem different, as it is not clear how to leverage
spanL or how to design a specific FPRAS.

\begin{definition}[Well-formed Strings]
    Let $L$ be a set of
    labels and let $oc \subseteq (L\times L)$ be an opening/closing relationship among labels.
    We define well-formed strings via a CFG $(\{S\}, L, P, S)$ where 
    $P=\{S\rightarrow \epsilon\} \cup \{S\rightarrow \alpha{}S \mid \alpha\in L, \not\exists (\alpha, \cdot)\in oc, \not\exists (\cdot, \alpha)\in oc\}\cup \{S\rightarrow \psi{}S\psi' \mid \psi\in L, (\psi, \psi')\in oc\}$.
\end{definition}

\begin{example}
    An example Dyck-like language would be $L=\{$``a'', ``('', ``['', ``)'', ``]''$\}$ and $oc=\{$(``('', ``)''), (```['', ``]'')$\}$. However, note that our variant is more general, as $oc$ could be any subset of $L \times L$.
\end{example}

Using well-formed strings, we consider the problem \#WFWalks, which counts edge-labeled $s$-$t$ walks in a graph, such that the labeling of any $s$-$t$ walk of length $n$ is a well-formed string. We thereby even support unlabeled edges (which can be seen as $\epsilon$-transitions in NFAs). %

\begin{definition}[Counting Well-formed Walks]\label{def:wfw}
    Given a graph $G=(V,E)$, vertices $s,t \in V$, a set $L$ of
    labels, an edge labeling $\varphi: E'\rightarrow L$ over a set $E'\subseteq E$ of edges, an opening/closing relationship $oc \subseteq (L\times L)$, and an integer $n$.
    Then, %
    the number of well-formed $s$-$t$ walks of size $n$ are defined as $|\{\varphi(P) \mid W\text{ is $s$-$t$ walk of length }n, \varphi(W)\text{ is well-formed string}\}|$, where %
    $\varphi((e_1,\dots,e_n))=\{w \mid \varphi(e_1) \text{ undefined}, w\in \varphi((e_2, \ldots, e_n))\} \cup \{\varphi(e_1)w \mid \varphi(e_1) \text{ defined}, w\in \varphi((e_2, \ldots, e_n))\}$.
\end{definition}

It is easy to see that while counting all $s$-$t$ walks of length $n$ is easy (\#L-complete, equivalent to \#DFA), counting labeled $s$-$t$ walks is \spn{}L-hard and
therefore not in FP (unless FP = \#P):

\begin{observation}
    Counting well-formed %
    walks %
    is \spn{}L-hard (under $\leq_{\mathsf{pars}\text{-}\log}$), even for $|L|{=}2$, $oc{=}\emptyset$.
\end{observation}%
\begin{proof}
    \#NFA over binary alphabet directly reduces to counting $s$-$t$ walks of length $n$.
\end{proof}

\begin{figure}
    \centering
    \includegraphics[width=0.35\linewidth]{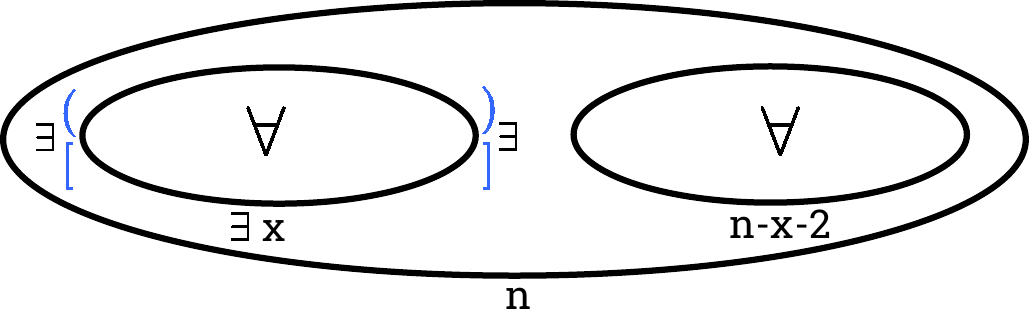}~\\[-1em]
    \caption{Visualization of how the ATrM for Dyck-like walks uses log space via universal branching.}%
    \label{fig:dyck}
\end{figure}

We \emph{establish an FPRAS for the general problem} by designing an ATrM, the idea of which is sketched in Figure~\ref{fig:dyck}.
Whenever we encounter a new opening label like ``('', we guess the length $x$ of this expression
as well as the edge with matching closing label ``)''. Then, we obtain two
subexpressions, one of length $x$ and one of length $n{-}x{-}2$. Both expressions need to be valid
and, therefore, form \emph{universal branches}. The full ATrM for computing well-formed walks is given in Algorithm~\ref{alg:wfwalks}.

\begin{algorithm}[t]\setstretch{0.85}
\small
\caption{\texttt{WFWalks}: ATM for computing well-formed walks on a graph}
\label{alg:wfwalks}
\SetKwFunction{WFW}{WFWalks}
\SetKwProg{Fn}{Function}{}{}
\KwIn{Graph $G=(V,E)$, labeling $\varphi$, open/closing $oc$, start vertex $s$, target vertex $t$, length $n$.}\medskip
\Fn{\WFW{$G, \varphi, oc, s, t, n$}}{
    \lIf*{$n = 0$}{\lIf*{$s=t$}{\textbf{Accept}}} \lElse*{\textbf{Reject}}%

    \tcp{$\exists$-branch: try a neighbor that induces well-formed labelings}%
    $u \leftarrow$ Guess $u$ with $e{=}(s,u)$, $e\in E$ \textbf{and} ($\varphi(e)$\text{ undefined} \textbf{or} $(\varphi(e),\cdot)\in oc$ \textbf{or} $(\cdot, \varphi(e))\notin oc$).

    \lIf*{no such $u$ exists}{\textbf{Reject}}\smallskip
    
            \lIf*{$\varphi(e)$ is defined}{\textbf{emit} $\varphi(e)$} %
            
        \If{$\varphi(e)$ is undefined \textbf{or} $(\varphi(e),\cdot)\notin oc$}
        {
            
            \lIf*{$\WFW(G, \varphi, oc, u, t, n{-}1)$ \textbf{Accepts}}{\textbf{Accept}}  %
            \lElse*{\textbf{Reject}}
        }
        \uElseIf{$(\varphi(e),\cdot)\in oc$} %
        { 
        
    \tcp{$\exists$-branch: try matching edges and size bounds, see Figure~\ref{fig:dyck}}
            $t' \leftarrow$ Guess vertex $t'$ with outgoing edge $e'$ to vertex $t''$ s.t.\ $(\varphi(e), \varphi(e'))\in oc$.
            
            $x \leftarrow$ Guess integer $x$ in $\{0, \dots, n{-}2\}$.

    \lIf*{no such $t'$ \textbf{or} no such $x$ exists}{\textbf{Reject}}
            
            \tcp{$\forall$-branch: recursively count walks, split along $t', t''$}
            \lIf*{$\WFW(G, \varphi, oc, u, t', x)$ \textbf{and} $\WFW(G, \varphi, oc, t'', t, n{-}x{-}2)$ both \textbf{Accept}}{\textbf{Accept}}  %
            \lElse*{\textbf{Reject}}
        }
}
\end{algorithm}%

Note that we can easily extend Algorithm~\ref{alg:wfwalks} so that we %
only consider balanced well-formed walks, where %
a constant number of labels 
appear (roughly) equally often along the path. This is interesting, since %
for deciding the word acceptance problem, context-free grammars are insufficient.

\begin{theorem}\label{thm:wfw}
    There is an FPRAS for counting well-formed $s$-$t$ walks of size $\leq n$ for some unary $n$.
\end{theorem}%
\begin{proof}[Proof Sketch.]
    Algorithm~\ref{alg:wfwalks} correctly characterizes well-formed $s$-$t$ walks, as defined in Definition~\ref{def:wfw}.
    We can show this by induction, where the base case is captured by Line 2. For the inductive step, 
    Line 3 guesses a neighboring edge $e$ where $\varphi(e)$ is opening or not closing.
    Line 7 captures the first inductive case where $\varphi(e)$ is neither opening nor closing, which inductively reduces to walk length $n{-}1$.
    Lines 9--12 cover the second inductive case (via universal branching, see also Figure~\ref{fig:dyck}), where $\varphi(e)$ is opening and we guess
    corresponding closing edge $e'$, which in Line 12 inductively reduces to two universal branches: 
    1) the expression of guessed size $x$ within the opening and closing labels,
    and 2) the remaining expression of size $n{-}x{-}2$ after the closing label.

    We sum up walks of length $0\leq i\leq n$ (see Corollary~\ref{lem:cfgfpras}) and 
    obtain an FPRAS by Thm.~\ref{thm:qpras}.
\end{proof}

\section{Conclusion and Outlook}
In this work, we advanced the understanding of when counting problems admit an FPRAS. We introduced a new machine model that extends the counting class spanL by adding universal (alternating) states to its underlying Turing machine model. This yields the new class \ourclass, which is strictly contained in \#P, yet every problem in \ourclass admits an FPRAS. Moreover, \ourclass has a natural complete problem: counting the number of accepted words of a context-free language up to a given length.
We demonstrated the usefulness of \ourclass by proving that certain families of graph query languages that require balanced labelings admit an FPRAS. We thereby push the boundary of demonstrating the existence of an FPRAS via simple algorithms via advanced counting machine models.
Our results so far are primarily of complexity-theoretic nature. The ATM-based characterization serves as a structural explanation for why these problems admit efficient approximation, rather than a ready-to-use FPRAS algorithm. An important direction for future work is to understand how to systematically extract practical approximation algorithms from \ourclass-machines, and to identify additional natural counting problems whose efficient approximation can be explained—and potentially improved—via this framework.
On the more complexity-theoretic side, several questions remain open. As emphasized in Conjecture~\ref{conj:spanlweaker}, we conjecture that \ourclass is strictly separated from spanL. Establishing such a separation is nontrivial, as it is linked to a longstanding open problem in counting complexity. 
Finally, the larger point of finding a computational characterization of precisely those problems that permit an FPRAS remains open.
Recall that Jerrum et al. showed in a landmark result that there is an FPRAS for the permanent of a matrix with non-negative values ($\mathsf{Per_{\ge 0}}$)~\cite{JerrumEtAl04}. This also forms the basis for a celebrated FPRAS result for bipartite perfect matching. While hard to formalize, there is a sense of these types of FPRAS results being fundamentally different from the flavor~of~FPRAS results derived by language theory (\#NFA, \#CFG, and corresponding TMs). Indeed, it is not believed that counting bipartite perfect matchings is in \spn{}L and we do not expect it to be in \ourclass either. While we believe these types of FPRAS are separate, we are not aware of any problem admitting an FPRAS that is in neither of the two classes. We pose this as an open problem to the~community.

\begin{openquestion}
Is there a problem that admits an FPRAS, but is not in $\ourclass \cup [\mathsf{Per_{\ge 0}}]^{\mathsf{aff}\text{-}P}$?
\end{openquestion}

\clearpage
\bibliographystyle{abbrv}
\bibliography{cite}

\appendix

\section{Additional Proof Details}

\subsection{Proof of Proposition \ref{prop:normal.atrm}}

Fix a witness ATrM $M$ for $P$, so $P(x)=\mathrm{span}_M(x;S,Z)$ for all $x$.
We build $M'$ from $M$ in two steps.

\paragraph{Step 1: make universal branching binary.}
Consider any universal configuration $C$ of $M$ (i.e., its control state is universal).
Let its set of successors be $\delta(C)=\{(C_1,\sigma_1),\ldots,(C_k,\sigma_k)\}$, ordered
according to the fixed child order $\prec$.
We transform this local fan-out as follows.

\begin{itemize}
\item If $k=0$, we relabel the corresponding state as \textsf{accept}.
This preserves acceptance semantics since a universal node with no children is vacuously accepting.
\item If $k=1$, we relabel the corresponding state as existential with the unique successor.
\item If $k\ge 2$, we replace $C$ by a a (linear length) sequence of binary branches:
we introduce fresh universal control states encoding an index $i\in\{1,\ldots,k-1\}$ and construct
configurations $C^{(i)}$ such that $C^{(1)}$ simulates $C$ and, for each $i<k-1$,
$C^{(i)}$ has exactly two successors:
one successor is $(C_i,\sigma_i)$ and the other is $C^{(i+1)}$ reached with output symbol $\circ$.
Finally, $C^{(k-1)}$ branches to $(C_{k-1},\sigma_{k-1})$ and $(C_k,\sigma_k)$.
\end{itemize}

Call the resulting machine $M_{\mathrm{bin}}$.
By construction, every universal configuration now has exactly two children.
Moreover, the transformation preserves outputs under  $\mathrm{Out}(\cdot)$:
all auxiliary edges introduced in the cascade emit $\circ$, hence the inserted auxiliary nodes are
$\circ$-labeled and therefore transparent under Algorithm~1 (they only splice together the output
forests of their children), see also \Cref{fig:binarization-output}. Thus, contracting each cascade yields a bijection between accepting
computation trees of $M$ and accepting computation trees of $M_{\mathrm{bin}}$, preserving
$\mathrm{Out}(T)$ up to isomorphism.

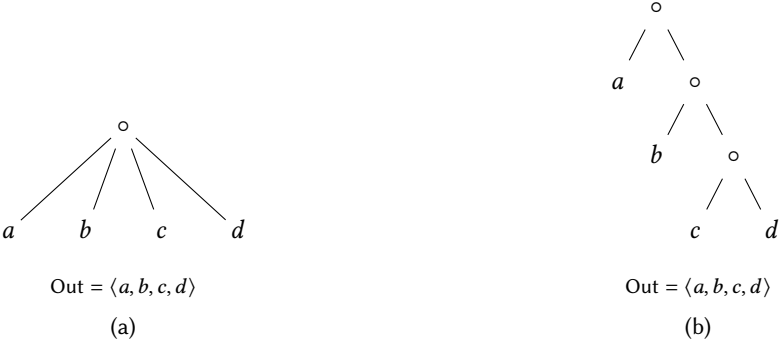
\begin{figure}[t]
\centering
\begin{subfigure}[t]{0.45\linewidth}
\centering
\begin{forest}
for tree={rounded corners, s sep=7mm, l sep=8mm, inner sep=2pt, align=center}
[$\circ$
  [$a$]
  [$b$]
  [$c$]
  [$d$]
]
\end{forest}

\vspace{0.3em}
{\footnotesize ${\sf Out}=\langle a,b,c,d\rangle$}
\caption{}
\end{subfigure}
\hfill
\begin{subfigure}[t]{0.45\linewidth}
\centering
\begin{forest}
for tree={rounded corners, s sep=7mm, l sep=4mm, inner sep=2pt, align=center}
[$\circ$
  [$a$]
  [$\circ$
    [$b$]
    [$\circ$
      [$c$]
      [$d$]
    ]
  ]
]
\end{forest}
\vspace{0.3em}

{\footnotesize ${\sf Out}=\langle a,b,c,d\rangle$}
\caption{}
\end{subfigure}
\caption{Binarizing a universal node preserves the output under the $\circ$-transparent semantics.}
\label{fig:binarization-output}
\end{figure}

Since $M$ is fixed, the maximum $k$ is a constant depending only on $M$, and therefore each
universal node is replaced by at most a constant number of auxiliary nodes. Hence there exists a
polynomial $Z_{\mathrm{bin}}$ with $Z_{\mathrm{bin}}(n)=O(Z(n))$ such that every accepting tree of size
$\le Z(n)$ in $M$ corresponds to an accepting tree of size $\le Z_{\mathrm{bin}}(n)$ in $M_{\mathrm{bin}}$,
and vice versa (after contraction). The work-tape space bound remains $O(S(n))$.
Consequently,
\[
\mathrm{span}_{M_{\mathrm{bin}}}(x;S, Z_{\mathrm{bin}})
\;=\;
\mathrm{span}_{M}(x;S,Z)
\quad\text{for all $x$.}
\]

\paragraph{Step 2: eliminate accepting computation trees beyond the size bound.}
The issue is that $M_{\mathrm{bin}}$ might also have accepting computation trees larger than
$Z_{\mathrm{bin}}(|x|)$, which are simply ignored by the definition of span with a finite size bound.
We show that we can always modify the machine so that such trees cannot be accepting at all.

Let $B=Z_{\mathrm{bin}}(|x|)$.
Since $Z\in\mathrm{poly}(B)$ is polynomially bounded in $|x|$, and thus can be represented in
$O(\log|x|)$ bits.
The machine $M'$ starts by computing $B$ on a separate track (or hardwiring the polynomial and
computing it from $|x|$), and then existentially guessing a budget $b\in\{1,\ldots,B\}$.
Intuitively, $b$ is the number of \emph{principal nodes} (i.e., simulated configurations of
$M_{\mathrm{bin}}$) allowed in the computation tree.

We then simulate $M_{\mathrm{bin}}$ while maintaining an integer counter (in binary) on the work tape:
each time we enter a principal simulated configuration, we decrement the counter by $1$; auxiliary
bookkeeping configurations do not increase the counter and are arranged so that they cannot loop and do not use universal states.
Formally, for each configuration $C$ of $M_{\mathrm{bin}}$ and counter value $c\ge 0$, $M'$ has a
simulation mode representing the pair $(C,c)$. Whenever $M'$ is about to simulate $C$ with counter
$c$, it first checks $c>0$; otherwise it rejects. It then decrements $c\leftarrow c-1$ and proceeds.

If $C$ is \textsf{accept} in $M_{\mathrm{bin}}$, then $M'$ accepts iff $c=0$ after the decrement,
and rejects otherwise.
 If $C$ is universal with exactly two successors $(C_1,\sigma_1)$ and $(C_2,\sigma_2)$,
then $M'$ first enters an auxiliary existential configuration that guesses nonnegative integers
$c_1,c_2$ with $c_1+c_2=c$, stores them, and then enters an auxiliary universal configuration that
branches to simulation modes $(C_1,c_1)$ and $(C_2,c_2)$, emitting $\sigma_1,\sigma_2$ respectively.
All transitions among auxiliary configurations emit $\circ$.

\smallskip
\noindent\emph{Correctness (span preservation).}
Fix an input $x$.

(\emph{$\subseteq$}) Let $T'$ be an accepting computation tree of $M'$ on $x$.
Project $T'$ to a tree $T$ by removing all auxiliary nodes (and forgetting the counters).
By construction, each principal node of $T'$ corresponds to exactly one node of $T$, and the
acceptance-by-budget condition (“accept iff the counter is $0$ after decrement”) implies that the
initial guessed budget equals the number of nodes of $T$.
In particular, $|T|\le B=Z_{\mathrm{bin}}(|x|)$.
Moreover, since auxiliary transitions emit $\circ$, Algorithm~1 ensures that
$\mathrm{Out}(T') \cong \mathrm{Out}(T)$.
Hence every output realized by $M'$ is realized by $M_{\mathrm{bin}}$ by some accepting tree of size
at most $Z_{\mathrm{bin}}(|x|)$.

(\emph{$\supseteq$}) Conversely, let $T$ be any accepting computation tree of $M_{\mathrm{bin}}$ on $x$
with $m:=|T|\le Z_{\mathrm{bin}}(|x|)=B$.
Consider the unique assignment of budgets to nodes of $T$ given by subtree sizes: assign to each
node the size of its subtree in $T$.
Using these values, $M'$ can guess the initial budget $b=m$ and, at each universal node, guess the
corresponding split into the two child subtree sizes.
This produces an accepting tree $T'$ of $M'$ whose projection is exactly $T$, and again
$\mathrm{Out}(T')\cong \mathrm{Out}(T)$ since all introduced nodes are $\circ$-labeled.
Therefore $M'$ realizes exactly the same output objects as $M_{\mathrm{bin}}$ within the size bound.

Combining both directions yields
\[
\mathrm{span}_{M'}(x;S',\infty) \;=\; \mathrm{span}_{M_{\mathrm{bin}}}(x;S,Z_{\mathrm{bin}})
\;=\; \mathrm{span}_{M}(x;S,Z).
\]

\smallskip
\noindent\emph{Bounding tree size and space for $M'$.}
By construction, $M'$ never accepts unless the principal simulated tree has at most $B$ nodes.
Between principal simulations it performs only bounded bookkeeping (arithmetic on $O(\log B)$-bit
counters, copying values, and moving between fixed auxiliary states), which cannot loop and thus
adds at most a polynomial factor overhead. Hence there exists $Z'(n)\in\poly(n)$ such that
every accepting computation tree of $M'$ on inputs of length $n$ has size at most $Z'(n)$, implying
$\spn_{M'}(x;S',Z')=\spn_{M'}(x;S',\infty)$.

Finally, $M'$ uses the work tape of $M_{\mathrm{bin}}$ plus $O(\log B)=O(\log n)$ extra bits to store
the counters, hence $S'(n)\in O(S(n)+\log n)$.
All universal configurations of $M'$ are binary by construction.

\subsection{Proof of Theorem~\ref{thm:cfg.eq.atm}}

\medskip
\noindent{\bf 1. $\boldsymbol{[\spn\alt(\log,\poly)]^{\mathsf{aff}\text{-}P}
=
[\#\text{CFG}]^{\mathsf{aff}\text{-}P}.}$}

\smallskip\noindent
\emph{($\subseteq$).}
Given $f\in \spn\alt(\log,\poly)$ and input $x$, Lemma~\ref{lem:atrm2cfg}
again yields a CFG $G$ with
\[
    f(x)=|\mathcal{L}_{\le t}(G)|,\qquad \|G\|\in\poly(|x|).
\]
Unlike the previous case, we do not require logspace preprocessing; the class
$[\cdot]^{\mathsf{aff}\text{-}P}$ allows polynomial-time affine 
transformations, and therefore converting $G$ into CNF or leaving it as an
arbitrary CFG are both permitted.  Hence, any $\spn\alt(\log,\poly)$ instance
can be mapped in polynomial time to a $\#\text{CFG}$ instance, giving the inclusion
$[\spn\alt(\log,\poly)]^{\mathsf{aff}\text{-}P} \subseteq [\#\text{CFG}]^{\mathsf{aff}\text{-}P}$.

\smallskip\noindent
\emph{($\supseteq$).}
Every CFG $G$ can be converted to $G_{\mathrm{CNF}}$ in polynomial time, so
\[
    [\#\text{CFG}]^{\mathsf{aff}\text{-}P}=[\#\text{CNFG}]^{\mathsf{aff}\text{-}P}.
\]
By Lemma~\ref{lem:cnfg2alt} we have $\#\text{CNFG} \le_{\mathsf{pars}\text{-}L}
\spn\alt(\log,\poly)$, hence in particular
$[\#\text{CNFG}]^{\mathsf{aff}\text{-}P}\subseteq [\spn\alt(\log,\poly)]^{\mathsf{aff}\text{-}P}$.
Combining these gives the desired equality.

\medskip
\noindent{\bf 2. $\boldsymbol{[\#\alt(\log,\poly)]^{\mathsf{aff}\text{-}P}
=
[\#\text{UCFG}]^{\mathsf{aff}\text{-}P}.}$}

\smallskip\noindent
\emph{($\supseteq$).}
Corollary~\ref{cor:ucfg} states a parsimonious polynomial-time reduction
\[
    \#\text{UCFG} \le_{\mathsf{pars}\text{-}P} \#\alt(\log,\poly).
\]
Hence, every $\#\text{UCFG}$ instance can be mapped in polynomial time to an
instance of $\#\alt(\log,\poly)$ while preserving the exact count; this gives
\[
    [\#\text{UCFG}]^{\mathsf{aff}\text{-}P}\subseteq[\#\alt(\log,\poly)]^{\mathsf{aff}\text{-}P}.
\]

\smallskip\noindent
\emph{($\subseteq$).}
Conversely, let $h\in\#\alt(\log,\poly)$ be realized by an alternating
logspace machine $M$ whose accepting computations we wish to count on input
$x$.  Apply Lemma~\ref{lem:atrm2cfg} to $M,x$ and the relevant bounds to obtain
a CFG $G$ such that
\[
    \#\text{acc}_M(x)=|\mathcal{L}_{\le t}(G)|.
\]
Moreover, the construction of Lemma~\ref{lem:atrm2cfg} %
is such that distinct accepting computations of $M$ give rise to distinct derivation trees
of $G$; thus the produced grammar $G$ is unambiguous for $h\in\#\alt(\log,\poly)$, i.e.\ $G\in\text{UCFG}$.
(The computational choices made along each accepting
computation is encoded into the derivation, preventing two different computations from
yielding the same derivation.)  Converting $G$ to CNF, if desired, is a
polynomial-time operation, so the whole mapping is an $\mathsf{aff}\text{-}P$
reduction to $\#\text{UCFG}$.  Therefore,
\[
    [\#\alt(\log,\poly)]^{\mathsf{aff}\text{-}P}\subseteq[\#\text{UCFG}]^{\mathsf{aff}\text{-}P}.
\]

\subsection{Proof of Theorem~\ref{thm:wfw}}
It suffices to show that Algorithm~\ref{alg:wfwalks} correctly characterizes \emph{exactly} the well-formed $s$--$t$ walks of length $n$, as in Definition~\ref{def:wfw}.
Indeed, assuming this correctness, we reduced the counting problem to the setting of Theorem~\ref{thm:qpras}, so we can just individually sum
over all lengths $\le n$, which works as in the proof of Corollary~\ref{lem:cfgfpras}.

\paragraph{Correctness of Algorithm~\ref{alg:wfwalks}.}
For each integer $m$ with $0 \le m \le n$, let $W_m$ denote the set of well-formed $s$--$t$ walks of
\emph{exact} size $m$ (Definition~\ref{def:wfw}).
We prove by induction on $m$ that:

\begin{quote}
\emph{Algorithm~\ref{alg:wfwalks} accepts precisely the walks in $W_m$ when started with budget $m$.}
\end{quote}

\medskip
\noindent\textbf{Base case $m=0$.}
Line~2 of Algorithm~\ref{alg:wfwalks} checks the start--target condition required by Definition~\ref{def:wfw} for an empty walk.
Thus Algorithm~\ref{alg:wfwalks} accepts exactly $W_0$.

\medskip
\noindent\textbf{Inductive hypothesis (IH).}
Fix $m>0$ and assume Algorithm~\ref{alg:wfwalks} accepts exactly $W_{m'}$ for all $m'<m$.

\medskip
\noindent\textbf{Inductive step.}
Consider a walk of size $m$.
When the algorithm is in some configuration $(v,m)$, Line~3 nondeterministically chooses an
outgoing edge $e=(v,u)$.
Write $\varphi(e)$ for its label.  
Exactly one of the following three cases holds:

\smallskip
\emph{Case 1: $\varphi(e)$ is neither opening nor closing.}
Algorithm~\ref{alg:wfwalks} then moves to configuration $(u,m-1)$ in Line~7.
By the IH, the recursive call accepts precisely the well-formed suffix walks of size $m-1$.
Hence this branch accepts exactly the walks whose first edge is $e$ and whose remaining part is
well-formed---i.e.\ exactly the walks in $W_m$ whose first edge is $e$.

\smallskip
\emph{Case 2: $\varphi(e)$ is a closing label.}
By Definition~\ref{def:wfw} a closing label cannot appear unless it matches a strictly earlier opening label.
Since the first step of the walk cannot be a matched closing, Algorithm~\ref{alg:wfwalks} rejects such branches.
This matches the fact that no walk in $W_m$ begins with a closing label.

\smallskip
\emph{Case 3: $\varphi(e)$ is an opening label.}
By well-formedness, every walk in $W_m$ that begins with $e$ has a \emph{unique} matching closing
edge $e'$ later in the walk.
This is the usual uniqueness of a matching bracket: the walk between $e$ and $e'$ contains the same number
of opening and closing labels of the same type, and $e'$ is the first such closing.
Let $x$ be the number of edges strictly between $e$ and $e'$.
Then any such walk $w$ decomposes \emph{uniquely} as:
\[
w = \varphi(e) \;\cdot\; w_{\mathsf{int}} \;\cdot\; \varphi(e') \;\cdot\; w_{\mathsf{suf}},
\]
where $|w_{\mathsf{int}}|=x$, $|w_{\mathsf{suf}}|=m-x-2$, and both $w_{\mathsf{int}}$ and $w_{\mathsf{suf}}$
are well-formed walks (Definition~\ref{def:wfw}).

Algorithm~\ref{alg:wfwalks} explicitly guesses such a matching $e'$ (Lines~9--12) and then
\emph{universally} branches into two recursive calls:
one on size $x$ for the interior, and one on size $m-x-2$ for the suffix.
By the IH, these two recursive calls accept exactly when $w_{\mathsf{int}}$ and $w_{\mathsf{suf}}$
are well-formed.
Universal branching enforces that both subcalls accept, hence this branch accepts
exactly the walks having first edge $e$ and matching closing edge $e'$ with well-formed
interior and suffix.
Because the matching closing edge $e'$ is unique, each such walk corresponds to
\emph{exactly one} accepting computation path of the algorithm, avoiding double counting.

\medskip

\noindent
Since the three cases are exhaustive and mutually exclusive, Algorithm~\ref{alg:wfwalks} accepts exactly $W_m$.
This completes the induction.

For showing the bijection, we still need to show completeness, i.e., we prove the following claim.

\begin{claim}
\label{lem:bijection-walks}
For every $m$, accepting computation paths of Algorithm~\ref{alg:wfwalks} with budget $m$ are in bijection
with well-formed walks in $W_m$.
\end{claim}

Soundness follows from the correctness above.
For completeness, given any walk in $W_m$, follow its unique decomposition at the first edge:
if non-bracket, as in Case~1 the next configuration is determined uniquely;
if opening, the matching $e'$ and the integer $x$ are uniquely determined by the standard bracket-matching rule.
Thus each walk induces a unique sequence of nondeterministic and universal choices in Algorithm~\ref{alg:wfwalks},
yielding a unique accepting path.

\section{On the Transducing Variant of $\mathsf{NAuxPDA}(\log, \poly)$}
\label{app:nauxpda}
\newcommand{\NAuxPDA}{\ensuremath{\mathsf{NAuxPDA}}\xspace}
We briefly spell out the transducer version of $\NAuxPDA(\log,\poly)$ alluded to
in \Cref{sec:atm} and justify the claimed reduction from $\spn\alt(\log,\poly)$ to this
model.

A \emph{transducing auxiliary pushdown automaton} (or
\emph{NAuxPDA-transducer}) is a nondeterministic AuxPDA as in \Cref{sec:logcfl}, equipped with a finite \emph{output alphabet}
$\Sigma_{\mathrm{out}}$ and a write-only, write-once output tape.  As for
ATrMs, transitions are labelled by symbols from
$\Sigma_{\mathrm{out}}^{\circ} := \Sigma_{\mathrm{out}} \cup \{\circ\}$,
where $\circ$ again denotes the “no-output’’ symbol.  Formally, if $\delta$
is the transition relation of the underlying AuxPDA, we assume that each
transition additionally carries a component
$\omega \in \Sigma_{\mathrm{out}}^{\circ}$, which is appended to the
output tape whenever the transition is taken (and ignored if
$\omega = \circ$).

Each accepting computation path $\rho$ of an \NAuxPDA $M$ on $x$ is assigned its content on the output tape as its \emph{output word}
${\sf Out}(\rho)$.
For functions $S,Z : \mathbb{N} \to \mathbb{N}$ we let
$ACC_M(x;S,Z)$ be the set of all accepting computations $\rho$ of $M$ on
input $x$ that use at most $S(|x|)$ work-tape space and at most $Z(|x|)$
steps.  We then define the \emph{span} of $M$ by
\[
  \spn_M(x;S,Z)
    \;:=\;
    \bigl|\{{\sf Out}(\rho) \mid \rho \in ACC_M(x;S,Z)\}\bigr|,
\]
i.e., the number of distinct outputs realized by accepting computations
within the given resource bounds.

\begin{definition}[The class $\spn\NAuxPDA(S,Z)$]
  For $S,Z : \mathbb{N} \to \mathbb{N}$, the class
  $\spn\NAuxPDA(S,Z)$ consists of all functions
  $f : \Sigma^{*} \to \mathbb{N}$ for which there exists an
  NAuxPDA-transducer $M$ such that
  \[
    f(x) \;=\; \spn_M(x;S,Z)
    \quad\text{for all } x \in \Sigma^{*}.
  \]
  In particular, $\spn\NAuxPDA(\log,\poly)$ is the counting analogue of
  $\NAuxPDA(\log,\poly)$.
\end{definition}

We next relate $\spn\alt(\log,\poly)$ to $\spn\NAuxPDA(\log,\poly)$.

\begin{lemma}\sloppy
\label{lem:naux.to.atrm}
  For every $f \in \spn\alt(\log,\poly)$ there exists an
  NAuxPDA-transducer $N$ with
  \[
    f(x)
      \;=\; \spn_N(x;\log,\poly)
    \quad\text{for all } x \in \Sigma^{*}.
  \]
  In particular, there exists a parsimonious log-space reduction from
  $\spn\alt(\log,\poly)$ to $\spn\NAuxPDA(\log,\poly)$.
\end{lemma}

\begin{proof}[Proof sketch]
  Let $M$ be an ATrM witnessing $f \in \spn\alt(\log,\poly)$, so
  $f(x) = \spn_M(x;S,Z)$ for some $S(n) = O(\log n)$ and $Z \in \poly(|x|)$.  By \Cref{prop:normal.atrm} we may assume w.l.o.g. that
  \begin{itemize}
    \item every universal configuration of $M$ has binary fan-out, and
    \item on input $x$, all accepting computation trees of $M$ have size
      at most $\poly(|x|)$ while $M$ uses only $O(\log |x|)$ work-tape
      space.
  \end{itemize}
We now describe an NAuxPDA-transducer $N$ that simulates $M$ in the
  standard depth-first manner.  Configurations of $M$ are encoded in the
  configuration of $N$ (control state, work tape, input head).
Existential states are simulated directly via non-determinism in $N$.
When a universal state is reached, $N$ places the current configuration on the stack and proceeds with the configuration of the first successor state (according to the $\prec$ order of the ATrM). Once an accepting state is reached, $N$ proceeds the simulation from the last configuration on the stack. When the simulation reaches a rejecting state of $M$, $N$ rejects also.
    
  Since $M$ is log-space and the depth of the computation tree is at
  most polynomial in $|x|$, the work tape of $N$ still uses
  $O(\log |x|)$ space and the overall running time is polynomial.

  We next explain how $N$ produces output.  First, whenever a transition
  of $M$ writes an output symbol $\sigma$, the corresponding transition
  of $N$  writes $\langle\sigma$ to its own output tape.  When the simulation reaches an accepting state of $M$, the machine $N$ write $\rangle$ to the output tape if the output tape of $N$ is non-empty.
Thus analogous to \Cref{lem:atrm2cfg} we encode the output of \Cref{alg:outT} as a string.

It is then immediate that each accepting computation tree $T$ of $M$ corresponds to exactly one accepting run $\rho_T$ of $N$ and vice versa. Moreover, ${\sf Out}(\rho_T)$ precisely encodes ${\sf Out}(T)$. Thus
  $\spn_N(x;\log,\poly) = \spn_M(x;S,Z) = f(x)$ for all $x$. The
  construction of $N$ from $M$ in logarithmic space is standard.
\end{proof}

\newcommand{\LOGCFL}{\ensuremath{\mathsf{LOGCFL}}\xspace}
\paragraph{On the Counting Limitations of the NAuxPDA--ATM Correspondence.}
The classical characterization
\[
\LOGCFL \;=\; \alt(\log,\poly) \;=\; \NAuxPDA(\log,\poly)
\]
relies on a simulation argument of Ruzzo~\cite{ruzzo1979tree} (itself based on a related reduction by Cook~\cite{cook1971characterizations}) that translates the behavior of an
auxiliary pushdown automaton into alternating computation.  Since this correspondence is central
to the conceptual placement of $\spn\alt(\log,\poly)$ in our framework, let us briefly recall the
structure of the simulation in the \emph{decision} setting and highlight why it cannot be adapted
to preserve counting information.

At a high level, the NAuxPDA to ATM reduction replaces concrete pushdown computations with a
\emph{surface-configuration} view: a configuration remembers the control state, work tape, and the
topmost stack symbol, but abstracts away the remainder of the pushdown stack.  The simulation
is then phrased in terms of a realizability predicate: a pair of surface configurations $(P,Q)$ is
realizable if there exists a valid NAuxPDA computation that starts in some full configuration
extending $P$ and ends in one extending $Q$, with properly matched push--pop behavior.  The ATM
decides this predicate recursively.  Its central step asserts that, whenever $P\neq Q$, then
$(P,Q)$ is realizable iff there exist surface configurations $R,S,T$ such that:
(i) $(P,R)$ is realizable,
(ii) $R \to S$ is a push transition,
(iii) $(S,T)$ is realizable,
and (iv) $T \to Q$ is a matching pop transition.
The ATM existentially guesses $(R,S,T)$ and universally verifies the subinstances.  Acceptance of
the input follows from realizability of the pair $(P_{\text{init}},Q_{\text{acc}})$.  

Crucially, this
argument merely witnesses \emph{existence} of some consistent run of the NAuxPDA. It does not
reference individual pushdown stacks or distinguish among different accepting runs.
From the standpoint of decision complexity this abstraction is innocuous: realizability is
well-defined, and the decomposition ensures that a polynomial-size alternating proof tree exists
precisely when the NAuxPDA accepts.  However, the same abstraction renders the reduction
fundamentally unsuitable for counting semantics.  Most notably, many distinct accepting
runs of the NAuxPDA may induce the \emph{same} surface path $(C_0,\ldots,C_\ell)$, as the
entire stack content below the top symbol is absent from the surface representation.  Hence, the
realizability predicate already collapses potentially exponentially many NAuxPDA runs into a
single surface witness.

One might attempt to refine the simulation to reduce nondeterminism in the decomposition.
For example, the ATM could existentially guess the length~$\ell$ of the realizing surface path
and, at every recursive call, universally branch only on the midpoint~$\lfloor \ell/2\rfloor$, thereby
enforcing a canonical subdivision of the underlying surface path.  This indeed eliminates the
proliferation of different proof trees \emph{for a fixed surface path} as each path admits at most one
consistent midpoint decomposition.  Nevertheless, the more fundamental collapse described above
persists: even a perfectly canonical decomposition can yield at most one accepting alternating
proof tree per surface path, while the NAuxPDA may have many accepting computations that all
project to the same $(C_0,\ldots,C_\ell)$.  Since these runs cannot be distinguished without
reinstating the full pushdown evolution, such a  alternation-based realizability proof can recover their
multiplicity.

There also exists a direct reduction from decision problems in $\NAuxPDA(\log,\poly)$ to plain PDAs (and hence CFG membership) \cite{sudborough1978tape}. Curiously, this reduction -- which is technically entirely different from the one described above, relying instead on reduction of two-way PDAs by one-way PDAs -- is also not counting-preserving.

Given \Cref{lem:naux.to.atrm} it thus seems that $\spn\NAuxPDA(\log,\poly)$ is a more powerful counting class than $\spn\alt(\log,\poly)$. Resolving whether a counting preserving reduction exists after all, or whether these classes are in fact separated is an interesting direction for future work. In particular, it is also open whether every problem in $\spn\NAuxPDA(\log,\poly)$  admits an FPRAS. 

 \end{document}